\RequirePackage[T1]{fontenc}
\documentclass[12pt]{article}

\usepackage[height=8.85in,width=6.25in]{geometry}

\usepackage[utf8]{inputenc}
\usepackage{amsmath}
\usepackage{amssymb}
\usepackage{mathtools}
\numberwithin{equation}{section}
\usepackage{slashed}
\usepackage{braket}
\usepackage[svgnames,dvipsnames]{xcolor}
\usepackage[colorlinks,citecolor=DarkGreen,linkcolor=FireBrick,urlcolor=FireBrick,linktocpage,breaklinks=true]{hyperref}
\usepackage{cite}
\usepackage{graphicx}
\usepackage{tikz}
\usepackage{pgfplots}
\pgfplotsset{compat=1.18}

\usepackage{times}
\usepackage{courier}
\usepackage{bm}
\usepackage{subfig}
\usepackage{here}

\usepackage{amsthm}
\theoremstyle{plain}
\newtheorem{theorem}{Theorem}
\newtheorem{lemma}{Lemma}

\newcommand{\comment}[1]{\textcolor{red}{[#1]}}

\def\bR{\mathbb{R}}
\def\bZ{\mathbb{Z}}
\def\cA{\mathcal{A}}
\def\cH{\mathcal{H}}
\def\Arf{\mathop{\mathrm{Arf}}\nolimits}
\def\supp{\mathop{\mathrm{supp}}}

\def\coc{\beta}
\def\com{\gamma}

\def\cF{\mathcal{F}}
\def\NS{\text{NS}}
\def\R{\text{R}}
\def\hotimes{\mathrel{\hat\otimes}}

\def\lambda{u}

\begin{document}

\begin{titlepage}

\begin{flushright}

\end{flushright}

\vskip 3cm

\begin{center}

{\Large \bfseries On anomalies and fermionic unitary operators}

\vskip 1cm
Masaki Okada$^1$,
Shutaro Shimamura$^2$,
Yuji Tachikawa$^1$
and
Yi Zhang$^1$
\vskip 1cm

\begin{tabular}{ll}
$^1$ & Kavli Institute for the Physics and Mathematics of the Universe (WPI), \\
& University of Tokyo,  Kashiwa, Chiba 277-8583, Japan, \\
$^2$ & Department of Physics,
 University of Tokyo,  Bunkyo, Tokyo 113-0033, Japan \\
\end{tabular}

\vskip 2cm

\end{center}

\noindent 
We point out that
fermionic unitary operators which anticommute among themselves appear in various situations
in quantum field theories with anomalies in the Hamiltonian formalism.
To illustrate, we give multiple derivations of the fact that position-dependent $U(1)$ transformations
of two-dimensional theories with $U(1)$ symmetry of odd level are fermionic
when the winding number is odd.
We then relate this mechanism to the anomalies of the discrete $\bZ_N\subset U(1)$ symmetry,
whose description also crucially uses unitary operators which are fermionic.
We also show that position-dependent $SU(2)$ transformations of
four-dimensional theories with $SU(2)$ symmetry with Witten anomaly
are fermionic and anticommute among themselves
when the winding number is odd.

\end{titlepage}

\setcounter{tocdepth}{2}
\tableofcontents

%%%%%%%%%%%%%%%%%%%%%%%%%%%%%%
\section{Introduction and summary}
\label{sec:intro}
%%%%%%%%%%%%%%%%%%%%%%%%%%%%%%

In recent years, we have seen significant improvements in our understanding of symmetries 
of quantum field theories (QFTs) and of their anomalies.
Among others, we have learned that allowed forms of anomalies of the same symmetry group in the same spacetime dimensions
can differ between bosonic QFTs and fermionic QFTs.
For example, consider $U(1)$ symmetry in two dimensions,\footnote{%
In this paper we use the hep-th convention of referring to the spacetime dimensions, rather than the spatial dimensions.
} whose anomalies are well-known to be characterized by an integer $k$ called its level.
In this case, any $k$ is allowed in fermionic QFTs, but $k$ has to be even in bosonic QFTs.
Similarly, in the case of $\bZ_2$ symmetry in two dimensions,
its anomaly is classified by $\bZ_2$ in bosonic QFTs but by $\bZ_8$ in fermionic QFTs,
such that $1\in \bZ_2$ maps to $4\in\bZ_8$.

The source of this difference can be understood in many different ways,
and in this paper we would like to understand it from the Hamiltonian point of view,
by studying how the symmetry operations are implemented in terms of unitary operators acting on the Hilbert space
of the theory.
Our central observation is the following. Given a theory in spacetime dimension $d$
with the Hilbert space $\cH$ for a spatial slice $M_{d-1}$,
we can consider unitary operators $U_1$, $U_2$ representing 
 symmetry operations localized respectively within disjoint regions $R_1, R_2 \subset M_{d-1}$,
 $R_1\cap R_2 =\varnothing$.
 In a \emph{bosonic} theory, the locality of the system dictates that they commute: \begin{equation}
 U_1U_2 = U_2 U_1,
\end{equation} whereas in a \emph{fermionic} theory, they can either commute or anticommute: \begin{equation}
U_1 U_2= \pm U_2 U_1.
\end{equation}
This additional possibility of having a minus sign here results in the difference of the structure of
anomalies in bosonic and fermionic QFTs, as we will amply see in this paper.

As the first example, which we present in Sec.~\ref{sec:2dU(1)},
we will discuss the case of $U(1)$ symmetry with level $k$ in two dimensions.
When we put the theory on a spatial circle\footnote{%
We prefer to use $S^1$ over $\bR$ as the spatial slice because of the following reasons.
On $\bR$, we can consider symmetry operations
which have different limits at two asymptotic infinites, $x\to\pm \infty$,
or do not have limits at all,
and they cause various additional complications.
} $S^1$, we can consider a position-dependent symmetry operation $U(f)$
acting on the Hilbert space $\cH$ schematically given by\footnote{%
In a  $G$-symmetric theory with dynamical fields $\phi(x)$, we can consider two types of position-dependent symmetry transformations given by functions $g:M\to G$:
\begin{itemize}
\item[A.] One is to do $\phi(x)\mapsto \tilde\phi(x):=g(x)\phi(x)$ and changing the background $G$-gauge field $A$ into $\tilde A(x):=g(x)Ag(x)^{-1}+ g(x) d g(x)^{-1}$.
\item[B.] Another is to do $\phi(x)\mapsto g(x)\phi(x)$, while keeping the background $G$-gauge field $A(x)$ unchanged,
one choice of which is to simply set $A=0$ throughout. 
\end{itemize}
The former simply gives a different description of the same system
whereas the latter can make an actual difference and does not necessarily commute with the Hamiltonian.
Denoting by $H[A]$ the Hamiltonian with the background gauge field $A$,
the former point of view is simply that $H[\tilde A]= U(g) H[A] U(g)^{-1}$ is conjugate to the original $H[A]$,
whereas the latter point of view emphasizes the fact that $H[\tilde A]$ is different from $H[A]$.
But both points of view deal with the same set of unitary operators $U(g)$ 
satisfying the same mathematical properties.

As the operation from the viewpoint B  does not commute with the Hamiltonian, 
one might say that it would not be appropriate to call it a `symmetry'.
But note that it is a common practice to call the entirety of $J_{\pm n}$ for a $U(1)$ current as an infinite-dimensional `symmetry',
and the operator $\exp(2\pi i a (J_n+J_{-n}))=\exp(2\pi i a\int_0^{2\pi} 2\cos(nx) J(x))=:U(2a\cos(nx))$ is 
an example of the type of unitary operators we will discuss in this paper.
\label{foot:fundamental}
} \begin{equation}
U(f) =\exp\left(2\pi i \int_{S^1} f(x)J^t(x) dx\right),
\label{Uf-intro}
\end{equation} where $f$ is a function on $S^1$ and $J^t(x)$ is the charge density operator.
The symmetry being $U(1)$ means that $f(x)$ is defined only modulo $1$,
so that it should be possible to define it for $f$ being a function from the spatial $S^1$ to the group manifold $U(1)$.
This brings various subtleties and complications in the definition of this operator in general,
and we are interested in the commutation relation among these operators $U(f)$.
This question was considered in the past by Segal and collaborators \cite{Segal:1981ap,MR797420,MR900587,Freed:2006ya,Freed:2006yc},
and our intention is to revisit this issue from a slightly different perspective.

Suppose $f$ is of winding number one and $\exp(2\pi if)$ is different from $1\in U(1)$ only in the region $R_1\subset S^1$.
Then we expect $U(f)$ to act locally in this region.
Similarly, suppose $g$ is also of winding number one and 
$\exp(2\pi ig)$ is different from $1\in U(1)$ only in the region  $R_2\subset S^1$
disjoint from $R_1$.
As the charge density operators $J^t(x)$ and $J^t(x')$  commute when $x$ and $x'$ are spatially separated,
we might naively expect that $U(f)$ and $U(g)$ commute.
We will see, however, that \begin{equation}
U(f) U(g) = - U(g) U(f)
\label{ac-intro}
\end{equation} when the level $k$ is odd. Such anticommutation is only allowed in a fermionic theory,
for which $U(f)$ and $U(g)$ can be fermionic and anticommute when supported on separate regions.
Stated differently, this sign forbids odd $k$ in a bosonic theory.
We give a derivation of this crucial sign \eqref{ac-intro} in two different ways, 
one by defining the operator \eqref{Uf-intro}
by a careful procedure, and another by a more axiomatic approach of starting from the 
commutation relations between winding-number-zero operators and extending them in a consistent manner.

As the second example, which we present in Sec.~\ref{sec:2dZ_n},
we consider the case of $\bZ_n$ symmetry, again in two dimensions.
In this case it is known that anomalies can be read off by considering the application 
of each element $g\in \bZ_n$ on a half-line $x>0$ but not on the other half $x<0$,
as was known in the algebraic QFT communities from decades ago e.g.~in \cite[Sec.~4.2]{Muger:2005ra}, 
and also in the condensed matter literature from the work of Else and Nayak \cite{Else:2014vma} from about ten years ago.
Denote by $\rho_g(O)$ the action of this half-space symmetry operation on operators $O$.
Then $\rho_g\rho_h$ should be basically equal to $\rho_{gh}$. But the application of a discrete symmetry
only on a half-line $x>0$ can have nontrivial effects around a small region $R$ around $x\sim 0$, 
leading to the condition \begin{equation} \label{eq:localunitary}
\rho_g(\rho_h(O))=u_{g,h}(\rho_{gh}(O)) u_{g,h}^{-1}
\end{equation} where $u_{g,h}$ is a unitary operator supported in this small region $R$.
Again, what distinguishes fermionic QFTs from bosonic QFTs is that 
$u_{g,h}$ can be fermionic in the former case.
We will review how this fact leads to a different classification of anomalies between fermionic and bosonic QFTs.
We also study how this anomaly arises when $\bZ_n$ is a subgroup of $U(1)$,
by constructing the operators $u_{g,h}$ explicitly in terms of $U(f)$ in Sec.~\ref{sec:2dU(1)}.

As the third and final example, which we present in Sec.~\ref{sec:4dSU(2)},
we will study the case of $SU(2)$ symmetry in four dimensions.
In that case, the anomalies are $\bZ_2$-valued, as originally found by Witten \cite{Witten:1982fp}.
Let us now consider unitary operators $U(f)$ implementing on the Hilbert space $\cH$
position-dependent $SU(2)$ transformations described by maps $f:\bR^3\to SU(2)$.
Suppose now two such maps $f_{1,2}:\bR^3\to SU(2)$ are nontrivial respectively 
only in disjoint compact regions $R_{1,2}\subset \bR^3$, $R_1\cap R_2=\varnothing$.
Again, naively we might expect that they should commute, but we will find \begin{equation}
U(f_1)U(f_2)=-U(f_2)U(f_1)
\end{equation} instead, when the winding numbers\footnote{
As we assume $f$ and $g$ are trivial outside of the compact regions $R_1$ and $R_2$,
we can regard them as maps from $S^3$ to $SU(2)$,
so the winding numbers are well-defined.
%When we speak of the winding number of $f:\bR^3\to SU(2)$,
%we always assume that the configuration of $f$ can be compactified at the point at infinity and can be regarded as $f:S^3\to SU(2)$.
%For a similar reason, we started from spatial $S^1$ in the two-dimensional case.
} of $f_1$ and $f_2$ are both one and when 
the $SU(2)$ symmetry in question has the Witten anomaly.
This result is analogous to the anticommutation \eqref{ac-intro} in the two-dimensional case,
but our derivation in the four-dimensional case is more abstract, using a geometric analysis
using the theory of invertible phases and $\eta$-invariants,
which also leads to a third and geometric derivation of \eqref{ac-intro}.

Before proceeding, the authors note that most of the results presented in this paper
are known in some form or another in the literature.
Firstly, the general philosophy of trying to see the manifestation of anomalies in the Hamiltonian framework
can be said to follow that of \cite{Delmastro:2021xox} but the implementation here is rather different.
More concretely, for the content in Sec.~\ref{sec:2dU(1)}, we already mentioned the work of Segal and collaborators \cite{Segal:1981ap,MR797420,MR900587,Freed:2006ya,Freed:2006yc}.
There is also a related work \cite{Bohm:1993yg} focusing on the cases of even level $k$.
For the content in Sec.~\ref{sec:2dZ_n}, we already noted the work of Else and Nayak \cite{Else:2014vma};
the work by Seifnashri \cite{Seifnashri:2023dpa} is also highly relevant.
And for the content in Sec.~\ref{sec:4dSU(2)}, the general idea is very much influenced by the works 
\cite{Freed:2006ya,Freed:2006yc}.
Also, the fact that unitary operators for winding number one $SU(2)$ transformations are fermionic 
can be thought of as almost already implicitly known from the work of Witten \cite{Witten:1983tx}, as we will comment in more detail later.
The content of a more recent paper by Jia and Yi \cite{Jia:2024kbl} is also closely connected.
That said,  the unified perspective emphasizing the role played by fermionic unitary operators 
in the diverse topics covered here was  useful and illuminating at least to the authors themselves,
and it was the desire of the authors to share it to the wider community that led them to produce this paper.

%%%%%%%%%%%%%%%%%%%%%%%%%%%%%%
\section{\texorpdfstring{$U(1)$}{U(1)} symmetry in two dimensions}
\label{sec:2dU(1)}
%%%%%%%%%%%%%%%%%%%%%%%%%%%%%%

%%%%%%%%%%%%%%%%%%%%%%%%%%%%%%
\subsection{The setup}
\label{sec:setup}
%%%%%%%%%%%%%%%%%%%%%%%%%%%%%%

We start our discussions by considering $U(1)$ symmetry in two dimensions.
On the spacetime of the form $\bR_\text{time}\times S^1_\text{space}$,
we normalize our charge density operator $J^t(x,t)$ to have the equal-time commutation relation
\begin{equation}
 [J^t(x,t), J^t(y,t)] =  k \frac{i}{2 \pi} \frac{\partial}{\partial y} \delta^{(\text{P})}(x - y)  \,.
 \label{eq:equal-time-commutator}
\end{equation}
We will parametrize the spatial $S^1$ by $[0,2\pi)$,
and the superscript $\text{P}$ in the delta function is to remind ourselves 
that it is a periodic delta function. 
For brevity, we drop the dependence on $t$ in the following.
Our charge density is normalized so that the position-independent $U(1)$ transformation is given by \begin{equation}
\exp\left(2\pi i\theta \int_{S^1} J^t(x) dx\right)
\end{equation} with the identification $\theta\sim\theta+1$.
We call the parameter $k$ appearing in \eqref{eq:equal-time-commutator} the level of the symmetry. It is normalized so that it has the value $k=1$ for a complex left-moving free fermion of charge $1$, 
as can be verified by an explicit computation.
We will see that $k$ is restricted to be an integer in a fermionic theory, and an even integer in a bosonic theory.

Let $LU(1) := \{f:S^1 \to  U(1) \}$ be the loop group of $U(1)$
and $LU(1)_0$ be its zero-winding-number subgroup that is connected to the identity.
Within the zero-winding-number sector, 
we consider the position dependent symmetry transformation $U(f_0)$ 
in terms of a function $f_0 : S^1\to \bR$, with the convention 
\begin{equation}
   LU(1)_0 \ni \exp(2 \pi i f_0) \longmapsto U(f_0) := \exp 2\pi i \int_{S^1} f_0(x)  J^t(x)\,dx \,.
   \label{eq:udef}
\end{equation}
Using the Baker-Campbell-Hausdorff (BCH) formula and \eqref{eq:equal-time-commutator},
we easily see
\begin{equation}
   U(f_0)  U(g_0) = \exp\left[ 2\pi i \coc_0(f_0,g_0) \right] U(f_0 + g_0) 
   = \exp \left[ 2\pi i  \com_0(f_0,g_0) \right]  U(g_0)  U(f_0) \,,
\end{equation}
where the $2$-cocycle map $\coc_0(f_0,g_0)$ and 
the commutator map\footnote{%
Here we are very slightly abusing the terminology.
Strictly speaking, it should be the exponentiated quantities, such as $\exp 2\pi i  \coc_0 $ and $\exp 2\pi i  \com_0 $, that define the cocycles and commutators on $LU(1)_0$.} $\com_0(f_0,g_0)$ on $LU(1)_0$
are given by\footnote{%
We note that higher-dimensional and higher-form versions of this formula were discussed in 
\cite{Hofman:2024oze,Vitouladitis:2025zoy}.
}
\begin{equation} \label{eq:B0andc0}
\begin{aligned}
\coc_0(f_0,g_0)&=  \frac{k}{2}   \int_{S^1} f_0 \,g'_0\,dx, \\
\com_0(f_0,g_0) &= \frac{k}{2} \int_{S^1} \left( f_0 g'_0 - g _0 f'_0 \right) dx  \,.\\
\end{aligned}
\end{equation}

What we would like to achieve in this section is to extend them to 
nonzero-winding-number sectors.
We will do so in two distinct ways, leading to the same answers.
In particular, we find that the commutation relation \begin{equation}
U(f) U(g)=\exp\left[2 \pi i \com(f,g)\right] U(g)U(f)
\end{equation} 
for the position-dependent $U(1)$ symmetry operators $U(f)$ and $U(g)$ is uniquely given by
the commutator map
\begin{equation}
    \com(f,g) = \frac{k}{2} \left( \int_0^{2\pi} \left( f(x) g'(x) - g (x) f'(x) \right) dx + f(0) w_g - w_f g(0)\right) 
 \,.
 \label{eq:full-result}
\end{equation}
Here, $f,g$ are functions $S^1\to U(1)$ represented as  maps $[0,2\pi]\to \bR$
with $w_f:=f(2\pi)-f(0)$ and $w_g:=g(2\pi)-g(0)$ being integers representing their winding numbers.

More concretely,
in Sec.~\ref{sec:1st-derivation},
we derive this result by giving a careful regularized definition of the operator $U(f)$,
by generalizing the operator $U(f_0)$ in the zero-winding-number sector 
which was given in \eqref{eq:udef}.
In contrast, in Sec.~\ref{sec:2nd-derivation},
we characterize the commutator map \eqref{eq:full-result} 
to be the unique consistent solution 
generalizing the zero-winding-number result \eqref{eq:B0andc0}
which satisfies the locality, i.e.~the condition \begin{equation}
U(f)U(g) =\pm U(g)U(f) \label{eq:gradedcommutation}
\end{equation} when $f$ and $g$ are supported on disjoint regions of the spatial $S^1$.
Two derivations given in Sec.~\ref{sec:1st-derivation} and Sec.~\ref{sec:2nd-derivation}
can be read independently.

Here are two immediate remarks. 
Using \eqref{eq:full-result}, we can easily compute the commutation relation
between the transformation  $U(\theta)$ by a constant $\theta$ 
and the operator $U(f)$ for the function with winding number $w_f=1$ is \begin{equation}
U(\theta) U(f) U(\theta)^{-1}= e^{2\pi i k\theta}U(f).
\end{equation}
As $U(1)=U(0)$, we find that the level $k$ should be an integer.
%\comment{YT: please make sure that the derivations in Sec~\ref{sec:1st-derivation}
%and Sec~\ref{sec:2nd-derivation} did not use the fact that $k$ is an integer at intermediate steps.
%If it is used in an essential manner, please tell me, as I need to change how to introduce
%the integrality of $k$ in the paper.}

We will find that the sign appearing in \eqref{eq:gradedcommutation} is dictated by the parity of $k$.
Using \eqref{eq:full-result}, we find that 
\begin{equation}
U(f)U(g) = (-1)^{k w_f w_g } U(g)U(f).
\end{equation}
This means that $U(f)$ is fermionic if and only if $k$ and $w_f$ are both odd.
This also means that the fermion parity operator $(-1)^F$ when $k$ is odd can be taken to be
\begin{equation}
(-1)^F = U(\frac12),
\end{equation} where $\frac12$ is a constant map sending all of $S^1$ to $\frac12$.

%%%%%%%%%%%%%%%%%%%%%%%%%%%%%%
\subsection{Comments}
\label{sec:sec-2-comments}
%%%%%%%%%%%%%%%%%%%%%%%%%%%%%%

Before proceeding to the two derivations,
let us give a couple of comments.

We begin with a comment on the history of this result \eqref{eq:full-result}.
To the authors' knowledge, this formula
or essentially equivalently a 2-cocycle for this commutator map \eqref{eq:full-result},
first appeared in the work by Segal and collaborators in the 1980s, see e.g.~\cite[\S 2]{Segal:1981ap}
and \cite[\S 4.7, Eq.\ (13.1.2)]{MR900587}.
There, the expression was derived as a unique solution generalizing the zero-winding-number result \eqref{eq:B0andc0} to the nonzero-winding-number sector,
satisfying a covariance under the action of diffeomorphisms on $S^1$.
Our derivation in Sec.~\ref{sec:2nd-derivation} is similar, but the set of 
additional conditions imposed are different.
We hope that our alternative derivations would shed different light on 
this important commutator map.

Next, we point out that the anticommutation $U(f)U(g)=-U(g)U(f)$ when $k=1$ is not so surprising
from one point of view, as already mentioned e.g.~in \cite{Else:2014vma}.
Very schematically, from the conservation of $J^\mu$,
we should be able to introduce a scalar operator $\chi$ such that $J^t=\partial_x \chi$ and $J^x=-\partial_t \chi$.
Then, given a winding-number-one function $f(x)$ such that $f(x)=0$ when $x<x_0$, $f(x)=1$ when $x>x_0$
and  a jump at $x=x_0$, we would naively have \begin{equation}
U(f)\, \text{``}{=}\text{''}\, \exp 2\pi i \int f(x)J^t(x) dx =
\exp 2\pi i \int f(x) \partial_x \chi(x) dx
= \exp (- 2\pi i \chi(x_0)). \label{eq:fermionizationformula}
\end{equation}
This is the standard fermionization formula,
which expresses the fermion operator, $\psi=e^{-2\pi i \chi}$,
in terms of the $U(1)$ current $J^\mu$ and the corresponding scalar field $\chi$.
Therefore, it is perfectly normal that these operators supported on different points to anticommute.
%But this relation might also be surprising (and actually was surprising to the authors),
%since the fermion operator $\psi$ is not usually thought of as (a limit of) a unitary operator.

Our discussion in Sec.~\ref{sec:1st-derivation} can be thought of as a more precise version 
of this quick and schematic derivation, which only applies to the case when the function $f(x)$ 
is a step function at $x=x_0$.
Our discussion also clarifies why and how this fermion operator $\psi$ can be thought of 
as (a limit of) \emph{unitary} operators corresponding to position-dependent symmetry transformations.
Actually, one important point of view of Segal and collaborators in \cite{MR900587}
was to give a mathematically rigorous definition of the 2d boson-fermion correspondence 
in terms of this step-function limit 
of the unitary operator for position-dependent symmetry transformation, which they called a \emph{blip}.

Finally, we would like to make a mathematical comment.
Let us set the level $k$ to be one.
Then, our commutator map  $\com(f,g)$, \eqref{eq:full-result}, is a function \begin{equation}
LU(1)\times LU(1)\to \bR/\bZ
\end{equation} which 
\begin{enumerate}
\item[(i)] defines a homomorphism for both variables $f$ and $g$,
\item[(ii)] satisfies $\com(f,g)=-\com(g,f)$, and
\item[(iii)] has the property $\com(f,f)=0$.
\end{enumerate}

Now, there is a concept in mathematics called differential cohomology $\hat H^d(X)$,
with a graded-commutative product $\hat H^d(X)\times \hat H^{d'}(X)\to \hat H^{d+d'}(X)$.
It was introduced by Cheeger and Simons in the 1980s \cite{CheegerSimons},
and is increasingly being used in the discussion of anomalies in the last several years in the hep-th community,
e.g.~in \cite{Cordova:2019jnf}.
Differential cohomology has the feature that \begin{equation}
\hat H^1(M_d)=\{ f: M_d\to S^1\}
\end{equation} and \begin{equation}
\hat H^{d+1}(M_d)=\bR/\bZ
\end{equation} where $M_d$ is any connected manifold of dimension $d$.
Taking $M_d=S^1$ in particular, we have $\hat H^1(S^1)=LU(1)$, and the anti-commuting product \begin{equation}
\hat H^1(S^1)\times \hat H^1(S^1) \to \hat H^2(S^1)
\end{equation} naturally gives a pairing $\tilde \com(f,g)$ \begin{equation}
\tilde \com: LU(1)\times LU(1)\to \bR/\bZ
\end{equation} which is also written as \begin{equation}
\tilde \com(f,g) = \int_{S^1} \hat f \hat g
\label{eq:gamma-tilde}
\end{equation} where $\hat f$ and $\hat g$ are functions $f,g$ viewed as elements of $\hat H^1(S^1)$.
This pairing
\begin{enumerate}
\item[(i)] defines a homomorphism for both variables $f$ and $g$, and
\item[(ii)] satisfies $\tilde \com(f,g)=-\tilde \com(g,f)$.
\end{enumerate}

These are very close to the properties of $\com(f,g)$ we want,
but $\tilde\com$ does not satisfy our condition (iii).
We need an explicit formula for  $\tilde \com(f,g)$ to see this.
The required formula was already given in Example 1.16 of the original paper \cite{CheegerSimons},
and there are various different-looking ways to present it. One such expression is
\begin{equation}
\tilde\com(f,g) = \int_0^{2\pi} f(x) g'(x) dx - w_f g(0),
\end{equation}
as explained in \cite{Cordova:2019jnf}.
A quick inspection shows that $\tilde \com(f,g)$ and $\com(f,g)$ (at $k=1$)  satisfy \begin{equation}
\tilde \com(f,g)=\com(f,g)+\frac12 w_f w_g. 
\label{eq:tilde-c-vs-c}
\end{equation}
In particular, \begin{equation}
\tilde \com(f,f)=1/2(w_f)^2
\end{equation}
which is indeed nonzero when the winding number $w_f$ is odd.
%This is compatible with the property (ii), which says $2\tilde \com(f,f)=0$ mod $1$, 
%which does not guarantee $\tilde \com(f,f)=0$ mod $1$ 
%and allows $\tilde \com(f,f)=0$ or $1/2$ mod $1$.

Note that $\com(f,g)$ and $\tilde \com(f,g)$ were introduced in \cite{Segal:1981ap, MR900587} and \cite{CheegerSimons} respectively,
both around 1985. It is unknown to the authors when their close connection was first realized;
at least it was mentioned in \cite{Freed:2006yc}.
We will see in Sec.~\ref{sec:4dSU(2)} how we can understand this relation using the general theory of invertible phases.

%%%%%%%%%%%%%%%%%%%%%%%%%%%%%%
\subsection{First derivation}
\label{sec:1st-derivation}
%%%%%%%%%%%%%%%%%%%%%%%%%%%%%%

Now that we have spent a couple of pages in the comments on the central result \eqref{eq:full-result},
let us move on to the derivations. The first derivation formulates the position-dependent symmetry operator by dividing space into patches. This formulation is the same as the one used for the time-dependent $\theta$-term in \cite{Cordova:2019jnf},
in which a physicist-friendly explanation of the differential cohomology product was given.

\paragraph{Mode expansion.} The standard mode expansion for the $U(1)$ current $J^t(x,t)$ is
\begin{equation}
\label{Jexp}
      J^t (x, t) =  \sum_{n \in \mathbb{Z}} \frac{1}{2\pi} J_n(t) \, e^{ i n x}\,.
\end{equation}
In the following, we omit the $t$-dependence of $J_n(t)$ and simply write $J_n$. The level-$k$ current algebra \eqref{eq:equal-time-commutator} is equivalent to the following equal-time commutation relation of $J_n$:
\begin{equation}
    [ J_n, J_m ] = k \,n \, \delta_{n + m, 0} \,.
\end{equation}
We define the primitive of $J^t(x)$ with respect to $x$, denoted by $\chi(x)$, such that $\partial_x\chi = J^t$, as
\begin{equation}
\label{chiexp}
    \chi (x, t) =  \sum_{n \in \mathbb{Z} \setminus \{0\}} \frac{1}{2\pi i n} J_n \, e^{ i n x } + \frac{1}{2\pi} J_0 \, x + \frac{1}{2\pi}I_0 \,.
\end{equation}
The extra mode $I_0$ obeys the commutation condition 
\begin{equation}
    [J_n,I_0] = ki\, \delta_{n,0}\,,
\end{equation}
so that the two new equal-time commutation relations are
\begin{align}
\label{com2}
    [J^t(x), \chi(y)] &=  k\frac{i}{2 \pi} \delta^{(\text{P})}(x - y) \,, \\
\label{com4}
    [\chi(x), \chi(y)] &=  k\frac{i}{2 \pi} \theta^{(\text{P}')} (x - y)  \,,
\end{align}
where the mode expansion of the periodic delta function $\delta^{(\text{P})}(x)$ is
\begin{align}
    \delta^{(\text{P})} (x) = \sum_{n \in \mathbb{Z} } \frac{1}{2\pi} e^{-i n x } \,,
\end{align}
and $\theta^{(\text{P}')}(x)$ is the quasi-periodic step function whose mode expansion is
\begin{align}
    \theta^{(\text{P}')} (x) = \sum_{n \in \mathbb{Z} \setminus \{0 \}  } \frac{i}{2\pi n} e^{-i n x } + \frac{1}{2\pi}x  \,.
\end{align}
The quasi-periodic step function satisfies 
\begin{align}
    \partial_x \theta^{(\text{P}')} (x) &= \delta^{(\text{P})} (x)  \,, \\
    \theta^{(\text{P}')} (x + 2\pi) &= \theta^{(\text{P}')} (x) + 1 \,,
\end{align}
and should be an odd function. This means that we have
\begin{align}
    \theta^{(\text{P}')}(x) = 
\begin{cases}
    m \,,        & x = m  \,, \\
    m + \frac{1}{2} \,, & 2\pi m < x < 2\pi (m+1) \,,
\end{cases}
\end{align}
for $m \in \mathbb{Z}$.

\paragraph{The regularized operator.}
%\comment{YT: please add a figure explaining the patching.}
The next ingredient that we need is the description of the spatial $S^1$ in terms of multiple closed intervals $\sigma_u \ (u = 1,2,...,n)$. 
We let $\sigma_{u,u+1}$ denote the points between adjacent intervals. 
We let $\sigma_{n,n+1}=\sigma_{0,1}$ be the same point between $\sigma_n$ and $\sigma_1$,
and let the value of $\chi$ jump there, i.e.~$\chi (\sigma_{n,n+1}) = \chi(\sigma_{0,1}) + J_0$. An illustration of the patching is shown in Fig.~\ref{fig:patch}.  %\footnote{To formulate this strictly, we should use partition of unity.} 

\begin{figure}
\centering
\includegraphics[width=0.6\hsize]{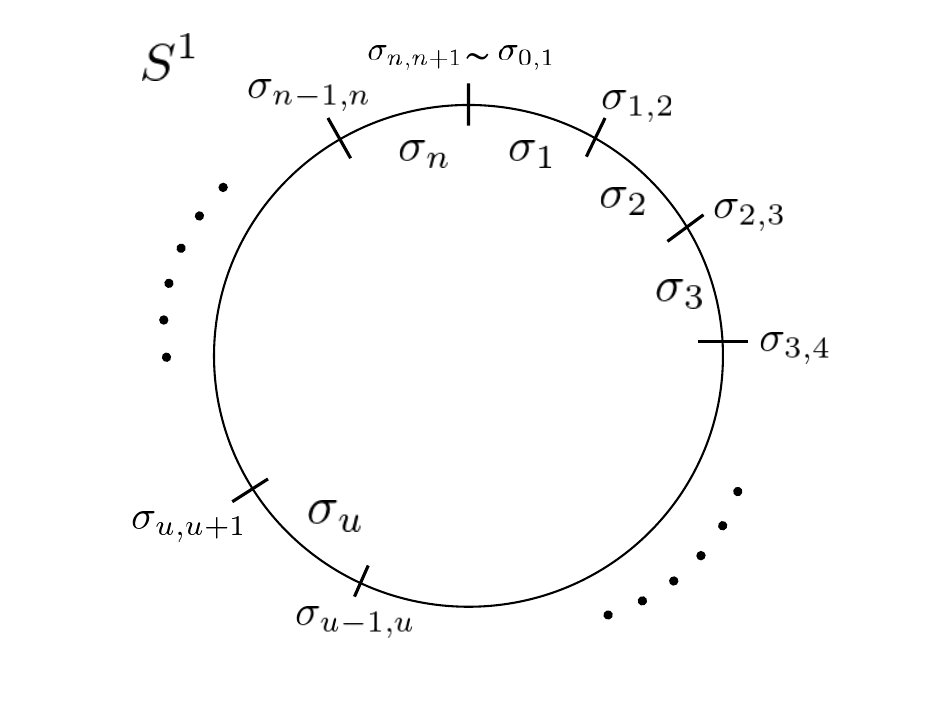}
\caption{The patching of $S^1$ with multiple intervals $\sigma_u$ and points $\sigma_{u,u+1}$. 
\label{fig:patch}}
\end{figure}

We define $f$ as a collection of functions
\begin{align}
    f = \left( \{f_u\}, \{w^f_{u,u+1}\}  \right) \,, 
\end{align}
where $f_u$ is a continuous and piecewise-smooth function $f_u:\sigma_u \rightarrow \mathbb{R}$ and $w^f_{u,u+1}$ are integers such that they satisfy
\begin{align}
    f_u(\sigma_{u,u+1}) - f_{u+1}(\sigma_{u,u+1}) = w^f_{u,u+1} \,.
\end{align}
For convenience, we let $f_{n+1} (\sigma_{n,n+1}) = f_1 (\sigma_{0,1})$ and $w^f_{n,n+1} = w^f_{0,1} = w^f_{n,1}$, so that we have $f_n (\sigma_{n,n+1}) - f_1(\sigma_{0,1}) = w^f_{n,1}$. The winding number of $f$ is given by $w_f =  \sum_{u=1}^n w^f_{u,u+1}$.

Then, we define the position-dependent symmetry operator $U(f)$ as
\begin{align}
\label{eq:trun_sym_op}
    U(f) = \exp \left[ 2\pi i \left( \sum_{u=1}^n \int_{\sigma_u} dx\, f_u(x) \partial_x \chi (x) - \sum_{u = 1}^n w^f_{u,u+1}\, \chi(\sigma_{u,u+1})  \right) \right] \,.
\end{align}
This formula is invariant under the gauge transformation
\begin{align}
    f_u &\sim f_u + m_u \,, \label{eq:gauge-trsf-of-patched-f}\\
    w^f_{u,u+1} &\sim w^f_{u,u+1} + m_u - m_{u+1} \,, \label{eq:gauge-trsf-of-patched-w}
\end{align}
where $m_u$ is an integer.\footnote{Under the gauge transformation \eqref{eq:gauge-trsf-of-patched-f}--\eqref{eq:gauge-trsf-of-patched-w}, the position-dependent symmetry operator $U(f)$ changes to $U(f)e^{2\pi i m_1 J_0}$.
Since eigenvalues of $J_0$ are integers, $U(f)$ is gauge invariant.} Using this transformation, we can change the choice of the partitions $\sigma_u$ freely.
In particular, we can take it to be composed of a single patch $\sigma_1=[0,2\pi]$.
%Furthermore, $U(f)$ is also invariant under the shift $\chi \rightarrow \chi + s \ (s \in \mathbb{Z})$, which allows us to freely choose the point where the value of $f$ as $ [\sigma_{0,1} = 0] \sim [\sigma_{1,2} = 2\pi] $, without loss of generality. 
With the above discussion, we arrive at a simplified version of $f$ described by 
\begin{align}
    f \colon [0,2\pi]\to \bR, \qquad
    f(2\pi) - f(0) = w_f \,,
\end{align}
with which $U(f)$ is given simply by\footnote{%
As already mentioned,
the exponent is the differential cohomology pairing when $\chi(x)$ and $J^t(x)=\partial_x\chi(x)$ are
$c$-numbers.
}
\begin{align}
\label{eq:trun_sym_op_single}
    U(f) = \exp \left[ 2\pi i \left(  \int_{\sigma_1} dx\, f(x)  J^t (x) - w_f\, \chi(0)  \right) \right] \,.
\end{align}
This form of $U(f)$ generalizes the fermionization formula \eqref{eq:fermionizationformula}, with the winding effect isolated in the second term. The function $f$ in the integrand now serves as a good test function for the periodic delta function, enabling us to compute the integration as in the zero-winding-number sector.

\paragraph{The 2-cocycle and the commutator.} 
To perform the computation to get the 2-cocycle and the commutator map, we need to take into account the normal ordering with respect to $J_n$ and define the normal ordered position-dependent symmetry operator $U_{\rm no}(f)$ as
\begin{align}
\label{eq:trun_sym_op_nord}
    U_{\text{no}}(f)  = {:U(f):} \,.
\end{align}
The  explicit form reads
\begin{align}
\begin{split}
\label{eq:trun_sym_op_nord_single}
    U_{\text{no}}(f) &={} :\exp \left[ 2\pi i  \left(  \int_{\sigma_1} dx\, f(x) J^t (x) - w_f\, \chi(0)  \right) \right] : \\
    &= \exp \left[ 2\pi i \sum_{s=1}^{\infty} \left(  \int_0^{2\pi} dx\, f(x) \frac{1}{2\pi} e^{-is x}  - w_f\, \frac{1}{2\pi i (-s)} e^{-is x}  \right) J_{-s} \right]  \\
    & \quad \times \exp \left[ 2\pi i  \left( - w_f \frac{1}{2\pi}  \right) I_0 \right] \\
    & \quad \times \exp \left[ 2\pi i  \left(  \int_0^{2\pi} dx\, f(x) \frac{1}{2\pi}   - w_f\,\frac{1}{2\pi} x  \right) J_0 \right] \\
    & \quad \times \exp \left[ 2\pi i \sum_{s=1}^{\infty} \left(  \int_0^{2\pi} dx\, f(x) \frac{1}{2\pi} e^{isx}  - w_f\, \frac{1}{2\pi i s} e^{isx}  \right) J_s \right]\,.
\end{split}
\end{align}
Then, the result of the 2-cocycle for $f$ and $g$ that have the same winding point can be computed using our commutation relations and the BCH formula:
\begin{align}
\begin{split}
    \coc(f,g) 
    =  -k & \left(   \int_{0}^{2\pi} dx\, \int_{0}^{2\pi} dy\, \partial_x f(x) g(y)  \sum_{s=1}^{\infty} \frac{1}{2\pi}   e^{is(x - y)} \right. \\
    & \left. \quad - \int_{0}^{2\pi} dx\,  f(x) w_g \sum_{s=0}^{\infty} \frac{1}{2\pi}  e^{isx}  -  w_f w_g  \sum_{s=1}^{\infty} \frac{i}{2\pi s} \right) \\
    = k & \left( \int_{0}^{2\pi} dx\, \int_{0}^{2\pi} dy\, f(x) 
    \partial_y g(y)  \sum_{s=1}^{\infty} \frac{1}{2\pi}    e^{is(x - y)} \right. \\
    & \left. \quad - \int_{0}^{2\pi} dy\, w_f g(y) \sum_{s=1}^{\infty} \frac{1}{2\pi} e^{-isy} + 2 \int_0^{2\pi} dx\, f(x)  w_g \frac{1}{2\pi}  +  w_f w_g  \sum_{s=1}^{\infty} \frac{i}{2\pi s} \right)
\end{split}
\label{eq:coc-by-patch-calc}
\end{align}
where we see the last term is divergent and regularization-dependent. Once we fix a regularization scheme
(an example of which will be explained in Appendix~\ref{sec:general-n-calc}),
we will get an identical term in $\coc(g,f)$. Taking the difference, we have the commutator map
\begin{equation}
 \begin{split}
    \com(f,g) &= \coc(f,g)  - \coc(g,f)  \\
        &= k\left(\int_{0}^{2\pi} dx\,  f(x) g'(x) - \int_{0}^{2\pi} dy \, w_f g(y) \, \delta^{(\text{P})}(y) \right) \\
        &= k\left(\int_{0}^{2\pi} dx\,  f(x)  g'(x) - w_f g(0) - \frac{1}{2} w_f w_g \right)\,,
\end{split}
 \end{equation}
where in the last step we used an endpoint regularization of the periodic delta function. Using integration by parts, we can rewrite the result into the form
\begin{equation}
\label{eq:commap_1}
     \com (f,g) 
    = \frac{k}{2} \left(\int_{0}^{2\pi} dx\, \left( f(x) 
     g'(x) - f'(x) \, g(x) \right)  + f(0) w_g - w_f  g(0) \right)  \,.
\end{equation}
Note that the 2-cocycle $\coc(f,g)$ and the commutator map $\com(f,g)$ are defined modulo $1$. The above results are for the case of $n=1$ in our language of $S^1$ patching, which reproduces \eqref{eq:full-result}. Alternatively, one can also ask for an expression for general $n$, as well as for the functions $f$ and $g$ whose windings are isolated at different points. 
We provide a detailed step-by-step explanation  in Appendix~\ref{sec:general-n-calc}.

%%%%%%%%%%%%%%%%%%%%%%%%%%%%%%
\subsection{Second derivation}
\label{sec:2nd-derivation}
%%%%%%%%%%%%%%%%%%%%%%%%%%%%%%

As mentioned above, there is another way to bypass the demand of an explicit form of $U(f)$ to obtain the explicit commutator map \eqref{eq:full-result} satisfying the graded commutativity \eqref{eq:gradedcommutation}.
For this, what we have to do is to extend the commutator map $\com_0$ (and the $2$-cocycle $\coc_0$ if needed) from $LU(1)_0$ to $LU(1)$.
We will omit the subscript $0$ to indicate quantities associated to the larger group $LU(1)$, e.g.\  $\com$ and $\coc$. 

\paragraph{The commutator map $\com$.}
To achieve the graded locality condition \eqref{eq:gradedcommutation},
only aiming at the commutator map $\com$ without considering the cocycle $\coc$ is sufficient.
Our starting point is the commutator map $\com_0$ given as in~\eqref{eq:B0andc0}.
Recall that an element of the zero-winding-number subgroup $LU(1)_0$ of the loop group $LU(1)$ is represented as $\exp(2\pi i f_0)$ using a continuous and piecewise-smooth function $f_0:S^1\to\bR$,
and the commutator map $\com_0$ in \eqref{eq:B0andc0} is a map $(f_0,g_0) \mapsto \com_0(f_0,g_0) \in \bR/\bZ$
such that it can induce a map $e^{2\pi i \com_0} : LU(1)_0 \times  LU(1)_0 \to U(1)$.
Similarly, an element of $LU(1)$ can be represented as $\exp(2\pi i f)$ using a function $f$ in
\begin{align}
\cF := \left\{f:\bR\to\bR \ \middle| \begin{array}{c} f(x+2\pi) = f(x) \bmod 1 \text{ for any $x$,} \\ \text{continuous and piecewise-smooth} \end{array} \right\}.
\end{align}
Now, we want a map $\com:\cF\times\cF \to \bR/\bZ$ satisfying the following conditions.
\begin{itemize}
\item[($\com$-0)] $\com$ is a map $\cF \times \cF \to \bR/\bZ$ such that it can induce a map $e^{2\pi i \com} : LU(1) \times LU(1) \to U(1)$.
That is,
\begin{align}
\com(f+1, g) = \com(f, g+1) = \com(f,g) \mod 1.
\end{align}

\item[($\com$-1)] $\com$ is bi-additive 
\begin{align}
\com(f+h,g)=\com(f,g)+\com(h,g), \quad \com(f,g+h)=\com(f,g)+\com(f,h) \mod 1,
\end{align}
and alternating
\begin{align}
\com(f,f)=0 \mod 1.
\end{align}
These are the conditions for $\com$ to be a commutator map.

\item[($\com$-2)] For $f_0,g_0\in \cF$ with winding number zero,
$\com$ reduces to $\com_0$ in \eqref{eq:B0andc0}. That is, we have
\begin{align}
\com(f_0,g_0) = \com_0(f_0,g_0)= & \ \frac{k}{2}\int_{S^1}(f_0(x)g_0'(x)-g_0(x)f_0'(x))dx \mod 1.
\end{align}

\item[($\com$-3)] $\com$ satisfies the graded locality condition\footnote{
We can actually weaken the condition ($\com$-3) to a condition ($\com$-3') which says that  
$\com(f,g)\in I \subset \bR/\bZ$, where $I$ contains $0$ as an isolated point, if $\supp f\cap\supp g=\varnothing$.
Then, the proof of Lemma \ref{lemma:wn-0-1-commutator} in Appendix~\ref{sec:proof-of-uniquness} is still valid.
As a result,
if this $I$ contains $\frac{1}{2}$,
then the uniqueness theorem \ref{thm:uniqueness-com} of the commutator map still follows,
and if not,
then the uniqueness theorem \ref{thm:uniqueness-com} modified so that $k$ is an even integer follows.
}
\begin{align}
\com(f,g) \in \frac{1}{2} \bZ \quad \text{if} \quad \supp f \cap \supp g = \varnothing.
\end{align}
Here, $\supp f := \overline{\{x\in [0,2\pi] \mid f(x) \neq 0 \bmod 1\}}$ for $f\in \cF$,
where $\overline{U}$ denotes the closure of the subset $U \subset \bR$.
\end{itemize}

Our central result is the following theorem\footnote{
In addition,
this unique commutator map $\com$ of Theorem \ref{thm:uniqueness-com} is invariant under any reparameterization of $S^1$.
We can see it in the course of the proof in Appendix \ref{sec:proof-of-uniquness},
but we can also check it explicitly as follows.
First, remember that we took the local coordinate of $S^1$ as a chart $(U=S^1\setminus\{pt\}, \, \phi:U\to(0,2\pi))$.
Any orientation-preserving change to another coordinate system $ \tilde{\phi}:U\to(a,b)$ is represented as the increasing bijection $\xi = \tilde{\phi} \circ \phi^{-1} :(0,2\pi) \rightarrow (a,b)$,
and it is obvious that
\begin{align*}
\com(f,g) & = \frac{k}{2} \left( \int_a^b (\tilde{f}(x)\tilde{g}'(x)-\tilde{g}(x)\tilde{f}'(x)) dx + \tilde{f}(a)w_g-w_f\tilde{g}(a) \right),
\end{align*}
where $\tilde{f}:=f\circ\xi^{-1}$.
Second, the shift of the starting point from $pt\in S^1$ to another point $\tilde{pt}\in S^1$ corresponds to the shift from $(0,2\pi)$ to $(p,p+2\pi)$ under the fixed coordinate system.
We can check by an explicit calculation that
\begin{align*}
\com(f,g) & = \frac{k}{2} \left( \int_{p}^{p+2\pi} (f(x)g'(x)-g(x)f'(x)) dx + f(p)w_g-w_fg(p) \right) \,.
\end{align*}
In short, this commutator map $\com$ is $\mathrm{Diff}^+(S^1)$-invariant.
This invariance property was taken up more prominently in \cite[\S2]{Segal:1981ap}.
\label{fn:reprm-inv}

%First, remember that we took the local coordinate of $S^1$ as a chart $(U_1=S^1\setminus\{p_1\}, \phi_1:U_1\to(0,2\pi))$. We can take another chart $(U_1, \phi_2:U_1\to(a,b))$, which is just a change of coordinate, and then we have
%\begin{align*}
%\com(f,g) & =  = \frac{k}{2} \left( \int_a^b (\tilde{f}(x)\tilde{g}'(x)-\tilde{g}(x)\tilde{f}'(x)) dx + \tilde{f}(a)w_g-w_f\tilde{g}(a) \right),
%\end{align*}
%where $\tilde{f}:=f\circ\phi_{12}^{-1}$ and $phi_{12}$ is the increasing bijection $\phi_{12} = \phi_2 \circ \phi_1^{-1} :(0,2\pi) \rightarrow (a,b)$. 
%Second, the commutator map is invariant under a shift of the starting point $0 \rightarrow p$:
%\begin{align*}
%\com(f,g) & = \frac{k}{2} \left( \int_{p}^{p+2\pi} (f(x)g'(x)-g(x)f'(x)) dx + f(p)w_g-w_fg(p) \right) \,.
%\end{align*}
%This can be checked by an explicit calculation.

%Or we can also see it in the course of the proof in Appendix \ref{sec:proof-of-uniquness}. 
%In short, this commutator map $\com$ is $\mathrm{Diff}^+(S^1)$-invariant.
%This invariance property was taken up more prominently in \cite[\S2]{Segal:1981ap}.
}
whose proof is in Appendix \ref{sec:proof-of-uniquness}:

\begin{theorem}
\label{thm:uniqueness-com}
If there is a commutator map $\com$ satisfying the consistency conditions {\rm ($\com$-0)--($\com$-3)},
then $k$ is an integer.
For any integer $k$, there is a unique such commutator map $\com$,
and its explicit formula can be given as
\begin{equation}
    \com(f,g) = \frac{k}{2} \left( \int_0^{2\pi} \left( f(x) g'(x) - g (x) f'(x) \right) dx + f(0) w_g - w_f g(0)\right) 
 \,. \label{eq:com-map-formula}
\end{equation}
\end{theorem}

\paragraph{The 2-cocycle $\coc$.}
If we are just interested in the commutator map, the question had already been settled in the previous paragraph.
We might be interested in the extension of the $2$-cocycle from the one $\coc_0$ of \eqref{eq:B0andc0} on $LU(1)_0$ to the one $\coc$ on $LU(1)$ as well.

From the general theory of group extension,
the cohomology classes of such $2$-cocycles $e^{2\pi i \coc}:LU(1) \times LU(1) \to U(1)$ are in one-to-one correspondence with the isomorphism classes of the group extensions of $LU(1)$ by $U(1)$.
Since we are considering the central extensions,
we can associate a commutator map $e^{2\pi i \com}$ by $\com(f,g) = \coc(f,g) - \coc(g,f)$ to each of the isomorphism classes of such extensions.

Conversely,
\if0
a theorem \cite[Prop.\ A.1]{Freed:2006ya}, \cite[Theorem 1]{Freed:2006yc}
applies to the central extensions of $LU(1)$ by $U(1)$,
and ensure that commutator maps (bi-multiplicative alternating maps $LU(1) \times LU(1) \to U(1)$) are in one-to-one correspondence with the isomorphism classes of such extensions.
\fi
as $LU(1)$ can be generated by $LU(1)_0$ and a single winding-number-1 function 
$z: S^1\to U(1)$, it is clear that the isomorphism class of the central extension formed by $U(f)$ for $f\in LU(1)$ 
is determined\footnote{
The details are as follows.
Suppose a commutator map $e^{2\pi i\com}:LU(1)\times LU(1)\to U(1)$ is given,
and a central extension $1 \to U(1) \to \widetilde{LU(1)} \to LU(1) \to 0$, whose restriction $1 \to U(1) \to \widetilde{LU(1)}_0 \to LU(1)_0 \to 0$ is already specified, has the commutator map $e^{2\pi i\com}$.
If we take a section $\tilde{U}:LU(1)\to\widetilde{LU(1)}$ so that $\tilde{U}(f) = \tilde{U}(f_0)\tilde{U}(z)^{w_f}$ for any $f=f_0+w_fz \in LU(1)$,
then we can see that the multiplication law of the central extension $\widetilde{LU(1)}$ is specified only by $\com$ as $\tilde{U}(f)\tilde{U}(g)=e^{2\pi i (\coc_0(f_0,g_0)+w_f\com(z,g_0))}\tilde{U}(f+g)$.
The 2-cocycle of any other section $U$ is in the same cohomology class as that of $\tilde{U}$.
} by the commutator map between $
U(z)$ and all the operators $U(f_0)$ for $f_0\in LU(1)_0$.
This concludes the uniqueness of the cohomology class of the $2$-cocycle $\coc$ we are looking for.

Let us state it in a form parallel to Theorem \ref{thm:uniqueness-com}.
We then start with the following conditions.
\begin{itemize}
\item[($\coc$-0)] $\coc$ is a map $\cF \times \cF \to \bR/\bZ$
such that it can induce a map $e^{2\pi i\coc}:LU(1) \times LU(1) \to U(1)$.
That is,
\begin{align}
\coc(f+1, g) = \coc(f, g+1) = \coc(f,g) \mod 1.
\end{align}

\item[($\coc$-1)] $\coc$ satisfies the cocycle condition
\begin{align}
\coc(g,h) - \coc(f+g,h) + \coc(f,g+h) - \coc(f,g) = 0 \mod 1.
\end{align}

\item[($\coc$-2)] For $f_0,g_0\in \cF$ with winding number zero,
$\coc$ reduces to $\coc_0$ in \eqref{eq:B0andc0}.
That is, we have
\begin{align}
\coc(f_0,g_0) = \frac{k}{2}\int_{S^1}f_0(x)g_0'(x)dx \mod 1.
\end{align}

\item[($\coc$-3)] $\coc$ is a $2$-cocycle for the commutator map $\com$ in (\ref{eq:com-map-formula}).
That is, $k$ is an integer and
\begin{align}
\coc(f,g)-\coc(g,f) = \com(f,g) \mod 1.
\end{align}
\end{itemize}

\begin{theorem}
There is a unique $2$-cocycle $\coc$ up to coboundary satisfying the conditions {\rm($\coc$-0)--($\coc$-3)}, and the explicit formula for a representative\footnote{
In the choice of a representative $\coc$,
there is a degree of freedom of adding a coboundary term.
This leads to some variations of the $2$-cocycles appearing in literature \cite[\S2]{Segal:1981ap}, \cite[\S 4.7, Eq.\ (13.1.2)]{MR900587}, and \cite[Eq.\ (4.2)]{Bohm:1993yg}.
The cocycle \eqref{eq:coc-by-patch-calc} we derived in the previous subsection is also another representative.
In particular, when the level $k$ is even,
we can choose $\coc$ so that it is $\mathrm{Diff}^+(S^1)$-invariant (see footnote \ref{fn:reprm-inv}).
For example, $\coc=\frac{1}{2}\com$ \cite[\S2]{Segal:1981ap}
and $\coc(f,g)=\frac{k}{2}\left(\int_0^{2\pi} f(x)g'(x)dx - w_f g(0)\right)$ \cite[Eq.\ (4.2)]{Bohm:1993yg} are $\mathrm{Diff}^+(S^1)$-invariant,
but satisfy ($\coc$-0) only when $k$ is even.
} can be given as
\begin{equation}
    \coc(f,g) =  \frac{k}{2}\left(\int_0^{2\pi} f(x)g'(x)dx + w_g f(0)\right)\,.
\end{equation}
\end{theorem}

%%%%%%%%%%%%%%%%%%%%%%%%%%%%%%
\section{\texorpdfstring{$\bZ_n$}{Zn} symmetry in two dimensions}
\label{sec:2dZ_n}
%%%%%%%%%%%%%%%%%%%%%%%%%%%%%%

In the last section, we discussed the projective phases associated to the position-dependent
$U(1)$ symmetry transformation in two dimensional quantum field theories.
In this section, we would like to study what we can learn from this 
about the anomalies of the cyclic group $\bZ_n$ by embedding it into $U(1)$.
To do this, we first need to recall the general theory of anomalies of finite group symmetry in two dimensions.

This section is organized as follows.
In Sec.~\ref{sec:finite-bosonic}, we review the standard story of 
how the anomalies of finite symmetry group $G$ 
can be extracted in the Hamiltonian formalism in two dimensions.
In Sec.~\ref{sec:2dfer}, we then extend this analysis to the fermionic theories in two dimensions.
This is also basically known, but we provide detailed discussions here as 
they are often left in the literature as exercises to the reader.
In Sec.~\ref{sec:embedding}, we apply the formalism reviewed in these two subsections
to the subgroup $\bZ_n\subset U(1)$.

\subsection{Anomaly of finite group \texorpdfstring{$G$}{G} in 2d bosonic theories}
\label{sec:finite-bosonic}

It is by now a common knowledge that the anomaly of a finite group $G$ in two-dimensional bosonic theories
is characterized by the group cohomology $H^3(BG;U(1))$.
This can be explained in various ways, but a Hamiltonian treatment is required for our purpose.
How to achieve this was long known in the algebraic quantum field theory community, see e.g.~\cite[Sec.~4.2]{Muger:2005ra}.
We present it following the argument of Else and Nayak~\cite{Else:2014vma}, who gave it in the context of spin chains, which we slightly adapt to the continuum language.

In the Hamiltonian picture of QFT, symmetry transformations acting on the states of Hilbert space $\cH$ are implemented by unitary operators $U_g$ 
corresponding to the action of $g$ on the entirety of the spatial slice without boundary.
In this case the assignment $g \mapsto U_g$ furnishes a (projective) representation of $G$,
i.e.~we have 
\begin{equation} \label{eq:repG}
    U_g U_h \propto  U_{g h} \,.
\end{equation}

To extract the anomaly of two-dimensional theories, 
we consider applying the symmetry operation $g$ on a subregion $W$ of the spatial slice.
Let $U_g$ denote the corresponding unitary operator.
Such manipulations would change the equation \eqref{eq:repG} to
\begin{equation} \label{eq:modifiedrepG}
    U_g U_h \propto  \,u_{g,h} U_{g h} \,,
\end{equation}
where $u_{g,h}$ is a unitary operator supported in the neighborhood of the boundary $\partial W$,
depending on how the jumps of the symmetry parameter at the boundary $\partial W$ are regularized.
These operators $u_{g,h}$ were called as \emph{fusion operators} in~\cite{Seifnashri:2023dpa}.

\begin{figure}[h]
\centering
\begin{tikzpicture}
\begin{axis}[
    hide y axis,
    axis x line=none,
    xmin=-1.5, xmax=12.5,
    ymin=0, ymax=3,
    width=13cm, height=5cm,
    clip=false,
    domain=-1.5:12.5,
    samples=200,
]

% Moderately reduced yellow region widths (not too narrow)
\addplot [yellow, opacity=0.4, fill=yellow, draw=none, forget plot] coordinates {(2.3,-1) (2.3,3) (3.2,3) (3.2,-1) (2.3,-1)};
\addplot [yellow, opacity=0.4, fill=yellow, draw=none, forget plot] coordinates {(8.8,-1) (8.8,3) (9.7,3) (9.7,-1) (8.8,-1)};

% Cosine bump function profiles
\addplot [thick, green!60!black, smooth, domain=-1.5:2.3] {0};
\addplot [thick, green!60!black, smooth, domain=2.3:3.2] {2.5 * (1 - cos(deg(pi*(x-2.3)/0.9)))/2};
\addplot [thick, green!60!black, smooth, domain=3.2:8.8] {2.5};
\addplot [thick, green!60!black, smooth, domain=8.8:9.7] {2.5 * (1 + cos(deg(pi*(x-8.8)/0.9)))/2};
\addplot [thick, green!60!black, smooth, domain=9.7:12.5] {0};
\node[anchor=west, green!60!black] at (axis cs:10.2,2.5) {\small$gh$};

\addplot [thick, blue, smooth, domain=-1.5:2.3] {0};
\addplot [thick, blue, smooth, domain=2.3:3.2] {2.0 * (1 - cos(deg(pi*(x-2.3)/0.9)))/2};
\addplot [thick, blue, smooth, domain=3.2:8.8] {2.0};
\addplot [thick, blue, smooth, domain=8.8:9.7] {2.0 * (1 + cos(deg(pi*(x-8.8)/0.9)))/2};
\addplot [thick, blue, smooth, domain=9.7:12.5] {0};
\node[anchor=west, blue] at (axis cs:10.2,2.0) {\small$h$};

\addplot [thick, red, smooth, domain=-1.5:2.3] {0};
\addplot [thick, red, smooth, domain=2.3:3.2] {1.5 * (1 - cos(deg(pi*(x-2.3)/0.9)))/2};
\addplot [thick, red, smooth, domain=3.2:8.8] {1.5};
\addplot [thick, red, smooth, domain=8.8:9.7] {1.5 * (1 + cos(deg(pi*(x-8.8)/0.9)))/2};
\addplot [thick, red, smooth, domain=9.7:12.5] {0};
\node[anchor=west, red] at (axis cs:10.2,1.5) {\small$g$};

% Horizontal axis
\draw[black, thin, ->] (-1.5,0) -- (12.5,0);

% Labels below yellow regions
\node[below] at (axis cs:3.25,0) {\small$\lambda^L_{g,h} \in \mathcal{A}^L$};
\node[below] at (axis cs:9.75,0) {\small$\lambda^R_{g,h} \in \mathcal{A}^R$};

\end{axis}
\end{tikzpicture}
\caption{Unitary fusion operators at the vicinities of boundary of supports.
\label{fig:fusionop}}
\end{figure}
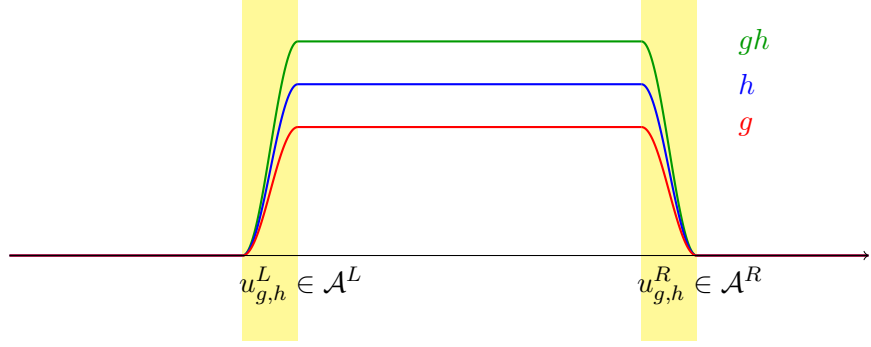

For definiteness, take $W$ to be a finite segment, see Fig.~\ref{fig:fusionop}.
We expect $U_g^W$, $U_h^W$ and $U_{gh}^W$ to be well-defined local operators  on $W$.
Then $u_{g,h}$ can be defined via \begin{equation}
U_g^W U_h^W = u_{g,h}^W U_{gh}^W. \label{XXX}
\end{equation}
$u_{g,h}^W$ should be a product of unitaries $u^L_{g,h}\in \cA^L$ and $u^R_{g,h}\in \cA^R$ localized 
on the left and the right boundaries of $W$, where $\cA^{L,R}$ are the algebras of operators
localized there.
We therefore have \begin{equation}
u_{g,h}^W = u^L_{g,h} u^R_{g,h}. \label{ZZZ}
\end{equation}
Note that this defines $u^{L,R}_{g,h}$ only up to a phase.

Now, let us define the action of $U_g$ on local operators via \begin{equation}
   O \longmapsto \rho_g(O) := U_g^W O (U_g^W)^{-1} \,.
   \label{YYY}
\end{equation}
From \eqref{XXX}, \eqref{ZZZ} and \eqref{YYY}, we have \begin{equation}
\rho_g \rho_h (O) =  \lambda_{g,h}^L \lambda_{g,h}^R \rho_{gh}(O) (\lambda_{g,h}^R)^{-1} (\lambda_{g,h}^L)^{-1} \, .
\label{WWW}
\end{equation}
Let us now suppose $O=O^L\in \cA^L$.
Then, from the commutativity of operators supported on disjoint regions,
we have 
    \begin{equation}  
\rho_g(\rho_h(O^L))=u^L_{g,h}(\rho_{gh}(O^L)) (u^L_{g,h})^{-1} \,.
\end{equation}

At this point onwards, we only need the left boundary of the region $W$
and the algebra of operators $\cA^L$ localized close to this boundary.\footnote{%
This also means that we can consider $W$ to be a semi-infinite half-line $x>0$.
If we do so, however, it is to be noted that the unitary $U_g^W$ is \emph{not} well-defined, 
although the automorphism $\rho_g \colon \cA^L \to \cA^L$ \emph{is} well-defined.
Assuming the existence of $U_g^W$ in this case runs into the following problem.
Namely, we would then have $U^W_g U^W_h= u_{g,h}^L  U^W_{gh}$.
Defining $u_{g,h}^L :=U^W_g U^W_h (U_{gh}^W)^{-1}$,
we can check that $\rho_g (u_{h,k}) u_{g,h k} = u_{g, h} u_{gh, k} = U^W_g U^W_h U^W_k (U^W_{ghk})^{-1}$
by an explicit computation.
This means that  $e^{2\pi i\alpha(g,h,k)}=1$ in \eqref{eq:Bosonkeyequality}
and there is no anomaly.
}
In the rest of this subsection, we will drop the superscript $L$, 
as every operator will be assumed to be from $\cA^L$.
We then have
\begin{equation}  \label{eq:localunitaryrepeat}
\rho_g(\rho_h(O))=u_{g,h}(\rho_{gh}(O)) (u_{g,h})^{-1} \,,
\end{equation}
which was \eqref{eq:localunitary} in the introduction.

Let us compute $\rho_g\rho_h\rho_k(O)$ in two ways.
First, we have \begin{equation}
\rho_g(\rho_h\rho_k)(O) =
\rho_g (u_{h,k} \rho_{hk}(O)u_{h,k}^{-1})
= [\rho_g(u_{h,k}) u_{g,hk}] \rho_{ghk}(O)[\rho_g(u_{h,k}) u_{g,hk}]^{-1}.
\end{equation}
Second, we have \begin{equation}
(\rho_g\rho_h)(\rho_k(O)) = u_{g,h}\rho_{gh}(\rho_k(O))u_{g,h}^{-1}
=[u_{g,h} u_{gh,k}] \rho_{ghk}(O) [u_{g,h} u_{gh,k}]^{-1}.
\end{equation}
Combined, we find that the combination \begin{equation}
[u_{g,h} u_{gh,k}]^{-1}  [\rho_g(u_{h,k}) u_{g,hk}] \in \cA \label{comb}
\end{equation} commutes with $\rho_{ghk}(O)$ for arbitrary $O\in \cA$.
As $\rho_{ghk}$ is an automorphism,
this means that the combination \eqref{comb} commutes with arbitrary elements in $\cA$.
We expect that the local algebra $\cA$ is such that the only element in $\cA$ 
which commutes with every element in $\cA$ is a scalar.\footnote{%
Such an operator algebra is called a \emph{factor}.
It is a standard assumption in algebraic quantum field theory that the algebra of local operators is a factor.
}
This means that there is a phase factor $\alpha(g,h,k)$ such that 
\begin{equation} \label{eq:Bosonkeyequality}
    \rho_g (u_{h,k}) u_{g,h k} = e^{2\pi i\alpha(g,h,k)} u_{g, h} u_{gh, k}\,.
\end{equation}

Next, one can show that $\alpha(g,h,k)$ is indeed a 3-cocycle~\cite{Else:2014vma} and its cohomology class $[\alpha]$ belongs to the group $H^3(BG, \bR/\bZ)$.\footnote{The Else–Nayak method for continuous group $U(1)$ was studied in~\cite{Kawagoe:2021gqi}, where a 3-cocycle (Eq.~(7.16)) was derived for bosonic theories.} We provide the computation using our notation in Appendix~\ref{subsec:cocycleconditionboson}.

As the fusion operator $u$ is defined only up to a phase, one has the freedom of introducing $\tau \in C^2(BG,\bR/\bZ)$ such that 
$\tilde u_{g,h} := e^{2\pi i\tau(g,h)} \, u_{g,h} $ and obtain a new phase factor $\tilde \alpha(g, h, k)$ through 
\begin{equation}
    \rho_g (\tilde u_{h,k}) \tilde u_{g,h k} = e^{2\pi i\tilde \alpha(g,h,k)} \tilde u_{g, h} \tilde u_{gh, k}\,.
\end{equation}
Inserting the relation between $u$ and $\tilde u$, then comparing with \eqref{eq:Bosonkeyequality}, we find
\begin{equation} \label{eq:bosonicequivalence}
\begin{aligned}
 \tilde \alpha(g,h,k) &= \alpha(g,h,k) + \tau(h,k)+\tau(g,hk)-\tau(g,h)-\tau(gh,k) \\
 &= \alpha(g,h,k)+ \delta \tau(g,h,k)\, \pmod \bZ,
\end{aligned}
\end{equation}
which means $[\tilde \alpha] = [\alpha]$ as  cohomology classes.

\subsection{Anomaly of finite group \texorpdfstring{$G$}{G} in 2d fermionic theories}
\label{sec:2dfer}

\paragraph{The data.}
When the theory is fermionic, the anomaly of finite group $G$ has a more intricate feature.
As was uncovered first in the condensed matter literature in \cite{Gu:2012ib}
and later studied from a mathematical perspective in \cite{Brumfiel:2016vpy},
 the anomaly consists of three data \begin{equation}
(\mu,\nu,\alpha)\in C^1(BG;\bZ_2)\times C^2(BG;\bZ_2)\times C^3(BG;\bR/\bZ),
\end{equation} where $C^d(BG;A)$ is the set of $A$-valued cochains of degree $d$, 
with the condition that \begin{equation}
\delta\mu=0,\quad
\delta\nu=0,\quad
\delta\alpha=\frac12 \nu^2.
\label{ferm-const}
\end{equation}
Therefore $\mu$ and $\nu$ are cocycles, while $\alpha$ is not a cocycle in general.
In this paper we consider the case $\mu=0$,
as nonzero $\mu$ never arises as a subgroup of $U(1)$.
Let us  review the Hamiltonian understanding of these data given in \cite{Else:2014vma}. 
As the derivation is somewhat long, we summarize here what we will find.
\begin{itemize}
\item For each $g\in G$, there is an automorphism $O\mapsto \rho_g(O)$ 
of the local operator algebra $\cA$.
\item For each pair $g,h\in G$, there is a unitary $u_{g,h}\in \cA$ such that \begin{equation}
\rho_g(\rho_h(O))=(-1)^{\nu(g,h)|O|}u_{g,h}\rho_{gh}(O)u_{g,h}^{-1},
\end{equation}
where $|O|$ is the fermion parity of the operator $O$
and $\nu(g,h)=|u_{g,h}|$.
\item Finally, we have phases $\alpha(g,h,k)$ such that \begin{equation}
   \rho_g (\lambda_{h,k}) \lambda_{g,hk} =
   e^{2\pi i\alpha(g,h,k)} 
   \lambda_{g,h} \lambda_{gh,k}\,.
\end{equation}
\item The data $\nu$ and $\alpha$ satisfy the constraints \eqref{ferm-const}.
\end{itemize}

\paragraph{The derivation.}
Our derivation starts in the same way as in the bosonic case,
by considering the operator $U_g^W$ 
for the symmetry operation $g$ on a finite segment $W$,
see Fig.~\ref{fig:fusionop}.
We have unitary operators $u_{g,h}^W$ via the relation \begin{equation}
U_g^W U_h^W = u_{g,h}^W U_{gh}^W, \label{eq:bosh}
\end{equation}
and  $u_{g,h}^W$  decomposes as \begin{equation}
u_{g,h}^W = u_{g,h}^L u_{g,h}^R,
\end{equation}
where $u_{g,h}^L\in \cA^L$ and $u_{g,h}^R\in \mathcal{A}^R$.
Then we have
\begin{equation}
\rho_g \rho_h (O) =  \lambda_{g,h}^L \lambda_{g,h}^R \rho_{gh}(O) (\lambda_{g,h}^R)^{-1} (\lambda_{g,h}^L)^{-1} \, \label{eq:foo}
\end{equation}
from \eqref{eq:bosh}, as before.

The crucial difference in the fermionic case is the following.
In general, operators $O^L\in \cA^L$ and $O^R\in \cA^R$  should either commute or anti-commute, \begin{equation}
O^L O^R = (-1)^{|O^L| |O^R|}O^R O^L,
\end{equation} 
where $|O|=0,1$ is its fermion parity.
Note that \emph{anti-commutation is impossible in a bosonic theory.} 

We expect $U_g^W$ for a finite segment $W$ to be bosonic,
as it can be continuously connected to identity by shrinking $W$.
Then $u_{g,h}^W$ should also be bosonic.
From this we see $|u_{g,h}^L|=|u_{g,h}^R|$.
Another consequence of $U_g^W$ being bosonic is that $|\rho_g(O)|=|O|$.

Now, set  $O=O^L\in \cA^L$ in \eqref{eq:foo}
and use \begin{equation}
\lambda_{g,h}^R O^L = (-1)^{|O^L||\lambda_{g,h}^R|} O^L \lambda_{g,h}^R.
\end{equation} 
We then have \begin{equation} 
\rho_g\rho_h(O^L) = (-1)^{|O^L||\lambda_{g,h}^R| }  \lambda_{g,h}^L \rho_{gh}(O^L)(\lambda_{g,h}^L)^{-1} \,.
\end{equation}
Using $|\lambda_{g,h}^R|=|\lambda_{g,h}^L|$,
we  conclude
\begin{equation} \label{eq:fermionfusion}
\rho_g\rho_h(O^L) = (-1)^{|O^L||\lambda_{g,h}^L| }  \lambda_{g,h}^L \rho_{gh}(O^L)(\lambda_{g,h}^L)^{-1} \,. 
\end{equation}
This is similar to \eqref{eq:localunitaryrepeat} in the bosonic case,
but has a sign factor.
From now on, we drop the superscript $L$ for brevity.

Let us now introduce a $\mathbb{Z}_2$-valued 2-cochain $\nu$ as
\begin{equation}
    \nu(g,h) :=|\lambda_{g,h}|.
\end{equation}
As before, we  compute $\rho_g \rho_h \rho_k(O)$ in two ways,
one as $   \rho_g (\rho_h \rho_k) (O) $ and another as  $(\rho_g \rho_h) \rho_k (O) $.
We find that \begin{equation}
\rho_g(\rho_h\rho_k)(O)=
   (-1)^{(\nu(h,k)+\nu(g,hk)) |O|}[\rho_g (\lambda_{h,k}) \lambda_{g,hk} ]
   \rho_{ghk}(O) 
   [\rho_g (\lambda_{h,k}) \lambda_{g,hk} ]^{-1}
   \label{AAA}
\end{equation} and \begin{equation}
(\rho_g\rho_h)(\rho_k(O))=
    (-1)^{(\nu(g,h)+\nu(gh,k))|O|}[\lambda_{g,h} \lambda_{gh,k}]
    \rho_{ghk}(O)
	[\lambda_{g,h} \lambda_{gh,k}]^{-1}.
\label{BBB}
\end{equation}

To go further, we need some reasonable assumptions on the structure of $\cA$.
We assume the following:
\begin{itemize}
\item It is graded, $\cA=\cA_0\oplus \cA_1$, where $\cA_0$ is bosonic and $\cA_1$ is fermionic.
\item It is \emph{not guaranteed} that there is a local version of the fermionic parity operator\footnote{%
Note that the global fermionic parity operator $(-1)^F$ is not necessarily an element of $\cA$.
We use the lower case $f$ for $(-1)^f$ to emphasize this point.
} 
$(-1)^f\in \cA$  such that $(-1)^f$ commutes with $\cA_0$ and anticommutes with $\cA_1$.
\item If an operator in $\cA$  commutes with every operator in $\cA_0$,
it is either a scalar, or a linear combination of a scalar and $(-1)^f$ if $(-1)^f$ exists.
\item If an operator in $\cA$  commutes with every operator in $\cA$,
it is a scalar. (This is a consequence of the last property.)
\end{itemize}
These properties are satisfied for example in a chain of Majorana fermion operators $\psi_i$.
In that case, $(-1)^f$ exists on the local operator algebra for an even number of sites,
but it does not exist on the local operator algebra for an odd number of sites.

Let us extract a consequence from the equality of \eqref{AAA} and \eqref{BBB}.
As $\rho_{ghk}(O)$ ranges over all $\cA_0$ when $O$ ranges over all elements of $\cA_0$,
we find that  \begin{equation}
[\lambda_{g,h} \lambda_{gh,k}]^{-1}
[\rho_g (\lambda_{h,k}) \lambda_{g,hk} ]
\label{CCC}
\end{equation} commutes with all elements of $\cA_0$.

At this point we separate the two cases, depending on whether $(-1)^f$ exists or not.
\begin{itemize}
\item
When $(-1)^f$ does not exist, the combination \eqref{CCC} should be a scalar, and therefore there should be 
a phase $\alpha(g,h,k)$ such that 
\begin{equation} 
   \rho_g (\lambda_{h,k}) \lambda_{g,hk} =
   e^{2\pi i\alpha(g,h,k)} 
   \lambda_{g,h} \lambda_{gh,k}\,. 
\end{equation}
Comparing the fermion parity of both sides, we have \begin{equation}
\nu(h,k) + \nu(g,hk)=\nu(gh,k)+\nu(g,h).
\end{equation}

\item
When $(-1)^f$ does exist, $(-1)^f\in \cA_0$ commutes with the combination \eqref{CCC}.
This means that \begin{equation}
\nu(h,k) + \nu(g,hk)=\nu(gh,k)+\nu(g,h).
\end{equation}
This implies that the sign prefactors in both \eqref{AAA} and \eqref{BBB}
are the same, and therefore the combination \eqref{CCC}
not only commutes with $\cA_0$ but with the entirety of $\cA$.
Therefore, the combination \eqref{CCC} is a scalar, and there should be a phase $\alpha(g,h,k)$ such that \begin{equation}
   \rho_g (\lambda_{h,k}) \lambda_{g,hk} =
   e^{2\pi i\alpha(g,h,k)} 
   \lambda_{g,h} \lambda_{gh,k}\,. 
\end{equation}
\end{itemize}

Either way, we have found that
\begin{equation}
   0 = \nu(h, k) - \nu(gh,k) + \nu(g,hk) - \nu(g,h) = \delta \nu(g,h,k)\,,
\end{equation}
showing that $\nu$ defines a 2-cocycle,
and that 
 there are phases $\alpha(g,h,k)$ such that 
\begin{equation} \label{eq:keyequality}
   \rho_g (\lambda_{h,k}) \lambda_{g,hk} =
   e^{2\pi i\alpha(g,h,k)} 
   \lambda_{g,h} \lambda_{gh,k}\,,
\end{equation}
%\begin{itemize}
    %\item 
    just as in  the bosonic version \eqref{eq:Bosonkeyequality}. 
    It is then straightforward to show that $\delta\alpha=\frac12 \nu^2$,
whose derivation we provide in Appendix~\ref{subsec:constraintfermion}.

\paragraph{The equivalence relation.} We have already seen in the bosonic case \eqref{eq:bosonicequivalence} that the phase ambiguity of the fusion operator translates into different choices of representatives of the bosonic anomaly cohomology class. 
An analogous argument based on the equation \eqref{eq:keyequality} shows that our 3-cochain $\alpha$ in the fermionic theory also admits the equivalence
\begin{equation} \label{eq:fermionalphaequivalence}
   \tilde \alpha(g,h,k) = \alpha(g,h,k) + \delta \tau(g,h,k)
\end{equation} 
as we can redefine
\begin{equation}
    \tilde\lambda_{g, h} = e^{2\pi i\tau(g,h)} \lambda_{g, h}
\end{equation} for $\tau(g,h) \in C^2(BG,\bR/\bZ)$.

However, this does not exhaust all the equivalent representations of $(\nu, \alpha)$. 
Recall that we defined $U_g^W$ as the action of the symmetry operation $g$ on a finite segment $W$.
Then nothing stops us from redefining
\begin{equation} \label{eq:bosonunitarymodi}
    \tilde U_g^W := \Sigma_g^W U_g^W \,,
\end{equation}
where $\Sigma_g^W$ is a unitary operator supported at the neighborhood of the boundary $\partial W$ of 
the segment $W$. 
We have \begin{equation}
\Sigma_g^W  = \Sigma_g^L \Sigma_g^R,
\end{equation} where $\Sigma_g^L\in \cA^L$ and $\Sigma_g^R\in \cA^R$,
and 
\begin{equation} \label{eq:fermionunitaryconj}
       \tilde\rho_g (O) =  \Sigma_g^L \Sigma_g^R \rho_g(O) (\Sigma_g^R)^{-1} (\Sigma_g^L)^{-1} \,.
\end{equation}
As argued before, we want $U_g$ and $\tilde U_g$ for finite segments to be bosonic.
This leads to
\begin{equation}
    |\Sigma^R_g| = |\Sigma^L_g| =: \xi(g) 
\end{equation}
with $\xi \in C^1(BG,\mathbb{Z}_2)$. 

Now suppose $O= O^L \in \cA^L$ in \eqref{eq:fermionunitaryconj}. 
We can use the graded commutativity to get 
\begin{equation} \label{eq:fermionicnewconj}
       \tilde\rho_g (O^L) = (-1)^{\xi(g) \,|O^L|} \Sigma_g^L \ \rho_g(O^L)  (\Sigma_g^L)^{-1} \,.
\end{equation}
We define the new fusion operator $\tilde \lambda^L_{g,h} \in \cA^L$ by 
\begin{equation} \label{eq:fermionicnewfusion}
\tilde\rho_g \tilde\rho_h(O^L) = (-1)^{|O^L|| \tilde\lambda_{g,h}^L| } \tilde \lambda_{g,h}^L \tilde\rho_{gh}(O^L)(\tilde\lambda_{g,h}^L)^{-1} \,.
\end{equation}
Using the equation \eqref{eq:fermionicnewconj} of $\tilde \rho_g$ and the identity  \eqref{eq:fermionfusion}, we derive from \eqref{eq:fermionicnewfusion} that 
\begin{align}
\tilde \lambda^L_{g,h} &=  \Sigma_g^L \rho_g(\Sigma_h^L) \lambda^L_{g,h}  (\Sigma_{gh}^L)^{-1}, \label{eq:utilde}\\
    \tilde \nu(g,h) := |\tilde\lambda_{g,h}^L|  &= \nu(g,h) + \xi(g) + \xi(h) - \xi(gh) =  \nu(g,h) + \delta \xi(g,h). \label{eq:nuequivalence}
\end{align}
We see that $\tilde \nu$ and $\nu$ differ by a coboundary term controlled by the fermion parity $\xi$ of $\Sigma^L$.

%An explicit computations as in the bosonic case  presented in Appendix B of~\cite{Else:2014vma} 
%shows that the 3-cocycle $\alpha$ remains the same.
 The next step is to find the expression of $\tilde \alpha$ specified by 
\begin{equation}  
   \tilde\rho_g (\tilde\lambda_{h,k}^L) \tilde\lambda_{g,hk}^L = e^{2\pi i\tilde\alpha(g,h,k)} \tilde\lambda_{g,h}^L \tilde\lambda_{gh,k}^L\,.
\end{equation}
A lengthy but straightforward computation (see Appendix~\ref{subsec:equivalencefermion}) shows that 
\begin{equation}
    \tilde\alpha(g,h,k) =\frac12\left( \xi(g) \left( \nu(h,k) + \xi(h) + \xi(k) - \xi(h k) \right) + \xi(k) \nu(g,h) \right) + \alpha(g,h,k)
\end{equation}
modulo $\bZ$.

We include the relation \eqref{eq:fermionalphaequivalence} and find the following equivalence relation between fermionic anomaly cochains $(\nu,\alpha) \in C^2(BG,\mathbb{Z}_2) \times C^3(BG,\bR/\bZ)$:
\begin{equation}
\begin{split}
   \nu(g,h) &\sim \nu(g,h) + \delta \xi(g,h),   \\
   \alpha(g,h,k) &\sim \frac12 \left(\xi(g)  (\nu(h,k)+\delta\xi(h,k))  + \xi(k) \nu(g,h) \right)+ \alpha(g,h,k) + \delta \tau(g,h,k) \,,
\end{split}
\end{equation}
for $(\xi,\tau) \in C^1(BG,\mathbb{Z}_2) \times C^2(BG,\bR/\bZ)$.

\paragraph{The addition formula.}
We note that the group operation
 for two anomaly data $(\nu,\alpha)$ and $(\nu',\alpha')$ is 
\begin{equation}
    (\nu,\alpha) \circ (\nu',\alpha') = (\nu + \nu', \frac12 \nu \cup_1 \nu' +  \alpha + \alpha') \,,
\end{equation}
where $\cup_1 : C^2(BG,\mathbb{Z}_2) \times C^2(BG,\mathbb{Z}_2) \rightarrow C^3(BG,\mathbb{Z}_2)$ is the cup-1 product~\cite{Brumfiel:2016vpy}, given concretely as
\begin{equation}
    (\nu \cup_1 \nu') (g,h,k) = \nu(g,h k) \nu'(h,k) + \nu(gh,k) \nu'(g,h) \,.
\end{equation}
This addition formula is derived in Appendix~\ref{subsec:supertensor}. 

\subsection{Anomaly of \texorpdfstring{$\bZ_n\subset U(1)$}{Zn in U(1)}}
\label{sec:embedding}

Our detailed understanding of the projective phase of position-dependent $U(1)$ transformations,
discussed in Sec.~\ref{sec:2dU(1)},
can be used to extract the data $(\nu,\alpha)$ describing the anomaly of the $\bZ_n$ subgroup of $U(1)$
in a very explicit fashion, through the implementation of the procedure given in the last subsection.
Let us carry it out. 

We specify first the profile function $\kappa(x)$ of the generator $\exp 2\pi i \frac{1}{n} \in \bZ_n \subset U(1)$
given by
\begin{equation}
\kappa(x)=\left\{\begin{array}{ll}
0 & (0 \leq x < a_1)  \\
\text{interpolate} & (a_1 \leq x < a_2)\\
\frac{1}{n} & (a_2 \leq x \leq a_3) \\
\text{interpolate}  & (a_3 < x \leq a_4) \\
0  & (a_4 < x \leq 2\pi)
\end{array}\right. \,,
\end{equation}
where $0 <a_1 < a_2 < a_3 < a_4 < 2\pi$.
For $a \in \{0,1,2,\ldots,n-1 \}$, we denote the labeling of group element $\exp 2\pi i \frac{a}{n}$ simply by $a$ and its profile is $a \kappa$. To display the fusion operators, we also introduce
\begin{equation}
f_L(x)=\left\{\begin{array}{ll}
0 & (0 \leq x < a_1)  \\
n \kappa(x) & (a_1 \leq x < a_2)\\
1 & (a_2 \leq x \leq 2\pi)
\end{array}\right., 
\quad \quad
f_R(x)=\left\{\begin{array}{ll}
0 & (0 \leq x < a_3)  \\
n \kappa(x) - 1& (a_3 \leq x < a_4)\\
-1 & (a_4 \leq x \leq 2\pi)
\end{array}\right..
\end{equation}
See Fig.~\ref{fig:profiles} for an illustration.

\begin{figure}
\[
\vcenter{\hbox{\includegraphics[width=.75\textwidth]{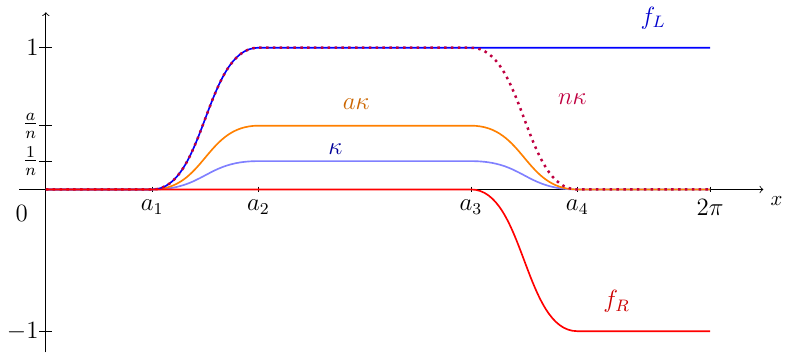}}}
\]
\caption{
The profile and auxiliary functions used in the analysis of the anomaly of $\bZ_n\subset U(1)$.
\label{fig:profiles}}
\end{figure}

Now define the carry\footnote{%
Note that this carry operation is an element in $H^2(\bZ_n;\bZ_n)$ for the extension 
$0\to \bZ_n \to \bZ_{n^2}\to \bZ_n\to 0$.
We all learn the case $n=10$ in the elementary school,
meaning that we all learned (a very basic case of) group cohomology there.
For more, see e.g.~\cite{IsaksenElementary}.
} $p(b,c)$ for $b,c \in \{0,1,2,\ldots,n-1 \}$ as
\begin{align}
p(b,c)=\left\{\begin{array}{ll}
0 & (b+c < n)  \\
1 & (b+c \geq n)
\end{array}\right.\,.
\end{align}
We also let $\overline{m}$  denote the residue of an integer $m$ modulo $n$.
When comparing  $\rho_b \rho_c (O)$ and $\rho_{\overline{b+c}} (O)$, 
nontrivial fusion operators appear if and only if $b+c \geq n$ and we then conclude that 
\begin{equation}
    \lambda^L_{b,c} = U(f_L)^{p(b,c)}\,, \quad \quad \lambda^R_{b,c} = U(f_R)^{p(b,c)} \,.
\end{equation}
We can compute 
\begin{equation}
    \begin{aligned}
         \rho_a(\lambda^L_{b,c}) = U(a\kappa)\left(U(f_L)^{p(b,c)}\right) U(a\kappa)^{-1}
         &= \left(\exp 2\pi i \,\com(a\kappa, f_L) \right)^{p(b,c)}\,U(f_L)^{p(b,c)} \\
         &= \left(\exp 2\pi i \frac{k}{2} \frac{a}{n}  p(b,c) \right)\,U(f_L)^{p(b,c)}\,.
\end{aligned}
\end{equation}
Putting the result back to the identity \eqref{eq:Bosonkeyequality}, we get 
\begin{equation}
    \left(\exp 2\pi i \frac{k}{2} \frac{a}{n}  p(b,c) \right)\, U(f_L)^{p(b,c)} U(f_L)^{p(a,\overline{b+c})}= 
    e^{2\pi i\alpha(a, b,c)} U(f_L)^{p(a,b)} U(f_L)^{p(\overline{a+b},c)} \,.
\end{equation}
 As the $p$'s are just $0$ or $1$, examining the above identity case by case, we get 
\begin{equation}
    \alpha(a, b,c) =  \frac{k}{2} \frac{a}{n}  p(b,c) 
    \,.
    \label{alpha}
\end{equation}
For the fermion parity we take
$(-1)^F = U(\frac{1}{2})$ with the constant map $g=\frac{1}{2}$,
see our discussion at the end of Sec.~\ref{sec:setup}.
Then   from
\begin{equation}
    (-1)^F \lambda^L_{b,c}(-1)^F = (-1)^{\nu(b,c)} \, \lambda^L_{b,c}
\end{equation}
we see 
\begin{equation}
  \nu(b,c)= k \,p(b,c)\,.
\end{equation}
It is straightforward to verify that 
\begin{equation}
    (\delta \alpha)(a, b, c, d) = \frac{k}{2} p(a,b) p(c,d)\,.
\end{equation}
This recovers 
\begin{equation}
    (\delta \alpha)(a, b, c, d) = \frac12 \nu(a,b) \nu(c,d) \,
\end{equation}
for odd $k$, while we have \begin{equation}
    (\delta \alpha)(a, b, c, d) = 0
\end{equation} for even $k$,
recovering the standard cocycle condition of the bosonic theory.
Indeed, Eq.~\eqref{alpha} for $k=2$
is the standard expression for the cocycle representative of the generator of \begin{equation}
H^3(\bZ_n;U(1))\simeq \bZ_n.
\end{equation}
%For $k=1$ and $n=2$, we get $\alpha:=\alpha(1,1,1) = i$ and $\nu:= (-1)^{\nu(1,1)} = -1$, so $(\alpha,\nu)= (  i, -1)$ which matches the number in~\cite{Else:2014vma}.
%We see that for $k=2$, $\alpha$ reduces to the familiar $3$-cocycle representative as in the bosonic case, while $(-1)^\nu \equiv 1$. 

%%%%%%%%%%%%%%%%%%%%%%%%%%%%%%
\section{\texorpdfstring{$SU(2)$}{SU(2)} symmetry in four dimensions}
\label{sec:4dSU(2)}
%%%%%%%%%%%%%%%%%%%%%%%%%%%%%%

In this last section\footnote{%
In this section, we use the viewpoint A in the footnote \ref{foot:fundamental} throughout,
in order to connect to the theory of invertible phases and anomalies.
}, we provide yet another method to obtain the
commutator map \eqref{eq:full-result} of the position-dependent $U(1)$ transformations
from the three-dimensional invertible phase encoding the anomaly of the two-dimensional boundary theory.
This alternative derivation can be readily generalized to the case of position-dependent $SU(2)$ transformations
of four-dimensional theories,
and we will find that two position-dependent $SU(2)$ transformations of winding number one
localized on two disjoint regions of the spatial $\bR^3$
anticommute with each other when the $SU(2)$ global symmetry has the Witten anomaly.

After giving reviews on anomalies of one-dimensional systems in Sec.~\ref{sec:1dim},
we will discuss $U(1)$ symmetry in two dimensions in Sec.~\ref{sec:XXX},
and $SU(2)$ symmetry in four dimensions in Sec.~\ref{sec:YYY}.
If the reader prefers, the last subsection \ref{sec:YYY} on $SU(2)$ in four dimensions
can be read without reading Sec.~\ref{sec:XXX} on $U(1)$ in two dimensions.
%if the reader is not interested in reproducing the results of Sec.~\ref{sec:2dU(1)} in this approach.

\subsection{Anomalies of one-dimensional systems}
\label{sec:1dim}
For our purposes, we need to review the anomalies of one-dimensional bosonic and fermionic systems,
and how they are encoded by the two-dimensional bulk invertible phase.

Consider a bosonic one-dimensional theory with symmetry $G$. 
This is simply a quantum mechanical system, where the group $G$ acts on the Hilbert space $\cH$.
It is non-anomalous when $\cH$ is a genuine representation of $G$,
and anomalous when $\cH$ is a projective representation of $G$.
It is well-known that a projective phase is classified in terms of $[\beta]\in H^2(BG;\bR/\bZ)$,
where the underlying 2-cocycle $\beta$ appears as the projective phase \begin{equation}
U_g U_h = e^{2\pi i\beta(g,h)} U_{gh}.
\label{eq:gen-proj}
\end{equation}

The same $[\beta]$ determines the bulk two-dimensional invertible theory,
whose action\footnote{%
In this section we somewhat abuse the terminology and represent
the partition function $Z$ of invertible phases as $Z=e^{2\pi i S}$ and 
call $S$ the action.
We consider $S$ to be valued in $\bR/\bZ$,
and any equality below involving $S$ should be regarded modulo $1$.
}  for a $G$-bundle on a manifold $M_2$ is given by \begin{equation}
\int_{M_2} f^*(\beta) \in \bR/\bZ,
\end{equation}
where $f:M_2\to BG$ is the classifying map for the $G$-bundle.
This action was first considered by Dijkgraaf and Witten in \cite{Dijkgraaf:1989pz},
although their motivation was to perform the finite path integral over all possible $G$-bundles.
Here we keep $G$-bundles as background fields.

More concretely, for $M_2=T^2$ with the $G$-bundle given by the commuting holonomies $g,h\in G$ along two 1-cycles,
the $U(1)$ phase is given by \begin{equation}
\com(g,h):=\int_{T^2} f^*(\beta)=\beta(g,h)-\beta(h,g).
\end{equation}  This is the discrete torsion first introduced by Vafa in \cite{Vafa:1986wx} whose relation to 
the group cohomology was noticed in \cite{Douglas:1998xa}.
The more immediate concern for us is that this is exactly the commutator phase appearing in \begin{equation}
U_g U_h = e^{2\pi i \com(g,h)}U_h U_g.
\end{equation}

We would now like to extend the discussion to the fermionic theories.
In this case, the anomaly is known to be characterized by the data \begin{equation}
([\mu],[\beta]) \in H^1(BG;\bZ_2) \times H^2(BG;\bR/\bZ),
\end{equation} which is a simpler analogue of the data for the anomalies of two-dimensional fermionic theories
we discussed in Sec.~\ref{sec:2dfer}.
In the Hamiltonian language, what distinguishes a fermionic theory is the existence of the fermionic parity operator $(-1)^F$
acting on the Hilbert space $\cH$.
The data $\beta$ controls the projective phase as in the bosonic case \eqref{eq:gen-proj}, 
while $\mu$ specifies the fermion parity \begin{equation}
U_g (-1)^F= (-1)^{\mu(g)} (-1)^F U_g.
\end{equation}
We note that the addition formula for the anomaly data is not simply the addition of cohomology classes but \begin{equation}
([\mu],[\beta]) + ([\mu'],[\beta'])
= ([\mu+\mu'],[\beta+\beta'+\frac12\mu\mu']),
\label{eq:1d-fer-add-formula}
\end{equation}
where the additional term $\frac12\mu\mu'$ appears 
from the fact that,
to combine the action of $U_g$ on $\cH$ and $U_g'$ on $\cH'$,
we have to use the super tensor product $U_g\hotimes U_g'$ on $\cH\otimes \cH'$,
see Appendix~\ref{subsec:supertensor}.

The same data $([\mu],[\beta])$ determines the bulk two-dimensional fermionic invertible theory,
whose action is \begin{equation}
S[M_2,\sigma, f] = \int_{M_2} f^*(\beta) + \frac12 q_{M_2,\sigma}(f^*(\mu)).
\end{equation} 
Here, we used $\sigma$ to denote the chosen spin structure on $M_2$,
and  $q_{M_2,\sigma}(a)\in \bZ_2$ for $a\in H^1(M_2;\bZ_2)$ is the quadratic refinement 
determined by the  spin structure $\sigma$ on $M_2$, given by  \begin{equation}
q_{M_2,\sigma}(\mu) = \Arf_{M_2}(\sigma+\mu) -\Arf_{M_2}(\sigma)
\end{equation} where $\Arf_{M_2}(\sigma)$ is the Arf invariant\footnote{%
The spin bordism group $\Omega^\text{spin}_2$ is $\bZ_2$,
and the Arf invariant is simply the bordism class of $M_2$ with the specified spin structure.
It is the number of zero modes of the Dirac operator on $M_2$ modulo two.
} of the manifold $M_2$ with the spin structure $\sigma$,
and $\sigma+\mu$ is the spin structure obtained by shifting the spin structure 
$\sigma$ by the $\bZ_2$ gauge field $\mu$.
For more details, see e.g.~\cite{Brumfiel:2016vpy,Kaidi:2019tyf}.
From this point of view, the extra term $\frac12\mu\mu'$ in the addition formula \eqref{eq:1d-fer-add-formula}
is due to the defining property of the quadratic refinement\footnote{%
\label{fn:quad-ref}%
Quadratic refinements serve the following purpose.
Say that we have a pairing $(a,b)$ valued in an Abelian group $\cA$.
\emph{If} it is possible to divide by $2$, the combination  $\frac12 (a,a)$
would satisfy the relation $$
\frac12(a+b,a+b) = \frac12(a,a)+\frac12(b,b) + (a,b).
$$
But when the pairing takes values in $\cA=\bZ_2$ or $U(1)$, we cannot divide by two 
in an unambiguous manner. 
In such a case, a function $q(a)$ satisfying the desired property $$
q(a+b)=q(a)+q(b) + (a,b)
$$ is called a quadratic refinement of the pairing $(a,b)$.
Usually, the data needed to  define  such a quadratic refinement are
more than the data required to define the original pairing.
For example, we need a spin structure to define a quadratic refinement.
}, \begin{equation}
q_{M_2,\sigma}(\mu+\mu')=q_{M_2,\sigma}(\mu)+q_{M_2,\sigma}(\mu')+\int_{M_2} \mu\mu'.
\label{eq:q-qua}
\end{equation}

Let us specialize this in the case of the torus, $M_2=T^2$.
On $T^2$, Arf invariant is nontrivial $1\in \bZ_2$ if and only if the fermion is periodic on all directions.
This allows us to compute $S[M_2,f]$ for $M_2=T^2$ and $f:T^2\to BG$ specified by two holonomies $(g,h)$,
which we denote by $S[T^2,\sigma,(g,h)]$.

Setting $h=e$ on $T^2$ with spin structure $(\NS,\R)$, we find
\begin{equation}
S[T^2,(\NS,\R), (g,e)] = \begin{cases}
0 &( \mu(g)=0),\\
\frac12 & (\mu(g)=1).
\end{cases}
\label{get-nu}
\end{equation}
In contrast, for $T^2$ with $(\NS,\NS)$ spin structure with holonomy $(g,h)$, the value is \begin{equation}
S[T^2,(\NS,\NS),(g,h)]=
\frac12\mu(g)\mu(h) + \com(g,h) .
\label{get-c}
\end{equation}
This consideration allows us to determine $\mu$ and $\com$ from \eqref{get-nu} and \eqref{get-c}
if we know how to evaluate the bulk invertible phase on the left hand side.

\subsection{Position-dependent \texorpdfstring{$U(1)$}{U(1)} transformation in two dimensions} 
\label{sec:XXX}

\subsubsection{The zero-winding-number sector}

We would like to set $G=LU(1)$ in the analysis above to read off $\com(g,h)$
for two-dimensional position-dependent $U(1)$ symmetry operations $U(g)$ and $U(h)$.
This can be done by compactifying the spatial direction of the two-dimensional theory on $S^1$
but \emph{not} throwing away the Kaluza-Klein mode,
and just formally regarding it as a one-dimensional theory.
As we will need multiple $S^1$'s, let us write this spatial $S^1$ in the original two-dimensional theory as $S^1_A$;
the temporal direction will be denoted by a subscript $X$.
Then the gauge group $G$ of the resulting one-dimensional theory is the group $LU(1)$ of functions from $S_A^1$ to $U(1)$.
The bulk invertible phase is then the 3d bulk invertible phase 
compactified formally on $S^1_A$,
and the bulk direction will be denoted by a subscript $Y$.

To compute the commutator map between $f:S^1_A\to U(1)$
and $g:S^1_A\to U(1)$, we need to have a $T^2=S^1_X\times S^1_Y$
such that the boundary condition around the direction of $S_X$ is twisted by $f$,
and that of $S_Y$ is twisted by $g$.
This defines a $G=LU(1)$ gauge configuration on $T^2=S^1_X\times S^1_y$, or equivalently, a $U(1)$ gauge configuration on $M_3:=S^1_X \times S^1_Y \times S^1_A$.
Then $\gamma(f,g)$ can be read off by evaluating the bulk invertible phase on this $M_3$.

The 3d invertible phase for $U(1)$ level $k$ is the Abelian Chern-Simons term,
which is often written as \begin{equation}
\exp(-2\pi i \frac {k}{2}\int_{M_3} \frac{A}{2\pi}   \frac{F}{2\pi})
\label{eq:expCS}
\end{equation} 
in physics literature.\footnote{%
The minus sign here is a convention to reproduce \eqref{eq:B0andc0} including the sign.
%In our convention, the level-$k$ $U(1)$ current algebra contributes to Chern-Simons term at level $-k$. This minus sign corresponds to the choice of the orientation of the boundary manifold $M_2$ relative to the bulk $M_3$.
}
This involves the gauge potential $A$, which is not in general globally well-defined,
and a somewhat subtle analysis is required to deal with such cases.
This will be detailed in the next subsection.
However, when $f$ and $g$ are both in the zero-winding-number sector,
a naive approach suffices, as $A$ can be thought of as a globally-well-defined one-form.
Let us perform this computation here and reproduce $\gamma_0(f,g)$ given in \eqref{eq:B0andc0}.

Let us parameterize $S^1_{X,Y,A}$ by $x,y,t$, all with periodicity $2\pi$.
Having the transformation $f: S^1_A\to U(1)$ around $S^1_X$ means that
the integral of the gauge field $A$ around $S^1_X$ should be equal to $2\pi f$;
recall that $\exp(2\pi if)$ takes values in $U(1)$ in our convention.
Similarly, the integral of $A$ around $S^1_Y$ should be equal to $2\pi g$.
One such $A$ is given by \begin{equation}
A=  f(t) dx +  g(t) dy.
\label{Azerow}
\end{equation}
We simply plug this in to \eqref{eq:expCS} and find that it is given by \begin{equation}
\exp(2\pi i \frac k2 \int_{S^1_A}  (f dg - gdf)),
\end{equation}
which indeed reproduces \eqref{eq:B0andc0}.

\subsubsection{The bulk invertible phase and the quadratic refinement}
\label{sec:sec-4-CSlevel}
We would now like to extend this computation to the sectors with
nonzero winding number.
For this, we need to make the Abelian Chern-Simons term \eqref{eq:expCS}
well-defined in the topologically nontrivial cases.
There are multiple ways to do so. \if0
We first review the one which is probably more common,
but unfortunately is not very useful in our computation.
We then provide another formulation,
which is more abstract but can be readily used for our present purposes.
\fi
We use two such equivalent but differently useful formulations.

\paragraph{Extending $M_3$ to $N_4$.}
One method of making it well-defined is the following.
From the bordism argument, for any $M_3$ with  a $U(1)$ bundle on it,
we can find a four-manifold $N_4$ whose boundary is $M_3$,
with the $U(1)$ bundle on $N_4$ extending that on $M_3$.
When $M_3$ has a chosen spin structure, we can arrange so that $N_4$ is also spin
and that the spin structure of $N_4$ restricts to the given spin structure on $M_3$.
Now, using the Stokes theorem, we try to define as follows: \begin{equation}
\int_{M_3} \frac{A}{2\pi}   \frac{F}{2\pi} \,\text{``}{=}\text{''}\, \int_{N_4}  (\frac{F}{2\pi})^2.
\label{eq:abelianCS}
\end{equation}
However there are multiple such choices of $N_4$ together with $U(1)$ bundle on it.
Say $N'_4$ is another such choice.
Two such tentative definitions would then differ by  \begin{equation}
\int_{N_4}  ( \frac{F}{2\pi}  )^2-
\int_{N'_4}  (\frac{F}{2\pi} )^2
=\int_{N_4 \cup_{M_3} \overline{N'_4}} (\frac{F}{2\pi})^2
=\int_{N_4 \cup_{M_3} \overline{N'_4}} c_1(F)^2.
\end{equation}
In the last equality, we used the fact that  $\frac{F}{2\pi}$ is the differential form representative 
of the first Chern class $c_1(F)$ of the closed manifold $N_4 \cup_{M_3} \overline{N'_4}$.

Since $c_1(F)$ is an integral cohomology class, the right hand side is then an integer.
For this statement we do not need  spin structure.
This makes the definition \eqref{eq:abelianCS} well-defined modulo $\bZ$.
We then find that 
the expression \eqref{eq:expCS} is well-defined  when $k$ is even, without using spin structure.

With spin structure, the intersection form on any four-manifold is known to be even,
guaranteeing that the right hand side is an even integer.
This makes the definition \eqref{eq:abelianCS} well-defined modulo $2\bZ$.
This means that  we need to specify the spin structure to make \eqref{eq:expCS} well-defined 
when $k$ is odd.
The dependence on the spin structure is the characteristic feature of a fermionic theory,
and we just saw that we need to have spin structure and therefore to have fermionic theory to realize odd $k$.
This point was emphasized e.g.~in \cite{Belov:2005ze}.
Our discussion in Sec.~\ref{sec:2dU(1)} was to understand this dependence on the parity on $k$ 
from the Hamiltonian point of view.

%This construction is not very useful in the actual computation, however.
\paragraph{Differential cohomology, quadratic refinement, and $\eta$ invariant.}
As another method, we use differential cohomology,
which we already briefly mentioned in Sec.~\ref{sec:sec-2-comments}.
For an introduction of differential cohomology for physicists, see e.g.~\cite{Cordova:2019jnf,Hsieh:2020jpj,GarciaEtxebarria:2024fuk}.

On a general manifold $M_d$ of dimension $d$, 
$\hat H^1(M_d)$ is the space of circle-valued functions $M_d\to S^1$,
$\hat H^2(M_d)$ is the space of $U(1)$ bundles with connections on $M_d$,
and $\hat H^{d+1}(M_d)$ is $\bR/\bZ$.
We also mentioned that there is a graded-commutative product on differential cohomology elements.
We then consider the gauge field $F=dA$ on $M_3$ as specifying an element $\hat A\in\hat H^2(M_3)$.
Taking two elements $\hat A_{1,2}\in \hat H^2(M_3)$,
we have $\hat A_1 \cdot \hat A_2\in \hat H^4(M_3)\simeq \bR/\bZ$.
Let us denote this value  by \begin{equation}
\int_{M_3}  \hat A_1 \cdot \hat A_2 =: (\hat A_1,\hat A_2) \in \bR/\bZ.
\end{equation}
When the $U(1)$ connections $\hat A_{1,2}$ are extended to $N_4$,
this product is known to be given by \begin{equation}
(\hat A_1,\hat A_2)=\int_{N_4} \frac{F_1}{2\pi} \frac{F_2}{2\pi}.\label{eq:A1A2}
\end{equation} In particular, when $\hat A_1=\hat A_2=\hat A$, the right hand side
reduces to the exponent of the Chern-Simons term \eqref{eq:expCS} at $k=2$
we discussed above.
We can then regard $(\hat A,\hat A)$ as a precise, well-defined version 
of the exponent of \eqref{eq:expCS} at $k=2$.\footnote{%
We note that the differential cohomology pairing, not only in this dimension but
more generally in arbitrary dimensions, can be defined 
directly on $M$ without introducing the bounding manifold $\partial N=M$,
by decomposing $M$ into patches. 
One of the simplest versions is the one used in our Sec.~\ref{sec:1st-derivation},
where we considered the pairing of $\hat f,\hat \chi\in \hat H^1(S^1)$,
although $\chi$ was operator-valued there.
}

To define the exponent of \eqref{eq:expCS} at $k=1$, 
we need to use the spin structure somewhere.
For this, recall that the definition \begin{equation}
\frac12 \int_{M_3} \frac{A}{2\pi}   \frac{F}{2\pi} 
=
\frac12\int_{N_4} (\frac{F}{2\pi})^2
=:Q_{M_3,\sigma}(\hat A)
\label{eq:QQ}
\end{equation} is well-defined modulo $\bZ$, in the presence of the spin structure $\sigma$.
Here we introduced the notation $Q_{M_3,\sigma}(\hat A)$ to denote this quantity.

Plugging in $A=A_1+A_2$ to this equation and expanding the right hand side,
we find \begin{align}
Q_{M_3,\sigma}(\hat A_1+\hat A_2)
&=Q_{M_3,\sigma}(\hat A_1)+ Q_{M_3,\sigma}(\hat A_2) + \int_{N_4} \frac{F_1}{2\pi}\frac{F_2}{2\pi}\\
&=Q_{M_3,\sigma}(\hat A_1)+ Q_{M_3,\sigma}(\hat A_2) + (\hat A_1,\hat A_2),
\label{eq:Q-qua}
\end{align}
where we used \eqref{eq:A1A2} here.
This relation identifies the $k=1$ Chern-Simons term, $Q_{M_3,\sigma}(\hat A)$,
as the quadratic refinement\footnote{See footnote~\ref{fn:quad-ref}.} 
of the differential cohomology pairing $(\hat A_1,\hat A_2)$.
Note the resemblance to \eqref{eq:q-qua} in the two-dimensional case,
where $q_{M_2,\sigma}(\mu)$ was a quadratic refinement for the pairing $\int_{M_2} \mu \mu'$.
This equation \eqref{eq:Q-qua} will turn out to be useful in our computation.

The quantity $Q_{M_3,\sigma}(\hat A)$ is also related to $\eta$ invariants.
Given the Dirac operator $\slashed{D}_A$ coupled to the $U(1)$ gauge field $A$
on a spin manifold $M_3$, the eta invariant $\eta(\slashed{D}_A)$ is the regularized 
sum of the signs of its eigenvalues $\lambda_a$: \begin{equation}
\eta(\slashed{D}_A) = \frac12 \text{``}\sum_a\text{''} \ \mathop{\mathrm{sign}} \lambda_a .
\end{equation} 
The Atiyah-Patodi-Singer index theorem says that \begin{equation}
\eta(\slashed{D}_A)=\int_{N_4} \left(\frac12(\frac{F}{2\pi})^2+\frac1{48}\mathop{\mathrm{tr}} (\frac{R}{2\pi})^2\right) \mod 1,
\end{equation} where $R$ is the spacetime curvature.
Comparing this with our definition \eqref{eq:QQ}, we find \begin{equation}
Q_{M_3,\sigma}(\hat A) = \eta(\slashed{D}_A)-\eta(\slashed{D}_{A=0}).
\label{eq:to-eta}
\end{equation}
The right hand side is defined solely in terms of $M_3$, without extending
$M_3$ to $N_4$.

\subsubsection{Computations}
We now want to proceed to the computation of the commutator map.
For this, we need to set up the three-dimensional manifold $M_3$
with a $U(1)$ bundle.
This manifold $M_3$ has a $T^2=S^1_X\times S^1_Y$ for which we apply the consideration of Sec.~\ref{sec:1dim}, 
and an $S^1_A$ on which we formally compactify the bulk theory.
Our $LU(1)$ is then the group of functions from $S^1_A$ to $U(1)$.
We take the periodicity of $S^1_{A,X,Y}$ to be $2\pi$.

We want the holonomy of $LU(1)$ along $S^1_X$ to be given by $f:S^1_A\to U(1)$.
This means that the curvature of the $U(1)$ connection on $S^1_X \times S^1_A$
is given by $\frac{dx}{2\pi}\times df$, where $x$ is the coordinate of $S^1_X$.
In the language of differential cohomology,
$\frac{dx}{2\pi}$ defines an element $\hat T_X \in \hat H^1(S^1_X)$ such that  \begin{equation}
\int_{S^1_X} \hat T_X=1,
\end{equation}
and is a differential lift of the generator of $H^1(S^1_X;\bZ)=\bZ$.
Similarly, we have $\hat f\in \hat H^1(S^1_A)$,
and the $U(1)$ connection on $S^1_X\times S^1_A$ is $\hat f \hat T_X$.
Analogously, we need a $U(1)$ connection on $S^1_Y\times S^1_A$ given by $\hat g\hat T_Y$.

In total, we have the $U(1)$ connection on 
$M_3=T^3 = S^1_A \times S^1_X \times S^1_Y$
given by  \begin{equation}
\hat A=\hat f \hat T_X + \hat  g\hat T_Y.\label{hatA}
\end{equation}
Here, $\hat f,\hat g\in LU(1)=\hat H^1(S^1_A)$ specify the gauge transformations used,
and $\hat T_X \in \hat H^1(S^1_X)$ and $\hat T_Y\in \hat H^1(S^1_Y)$
are differential lifts of the generators of  $H^1(S^1_{X,Y};\bZ)=\bZ$.
Compare this with the gauge field \eqref{Azerow} in the zero winding number sector;
the expression \eqref{hatA} is simply the version of \eqref{Azerow}
applicable in the general, topologically non-trivial cases.

The value of the invertible phase 
depends on the spin structure $\sigma$ on $T^3$, so let us specify it by defining
$a_i$ for each direction $i=A,X,Y$ of $S^1$
to be $0$ or $1$ depending on whether $S^1_i$ is in the NS sector or in the R sector.
Recalling \eqref{eq:expCS}, our objective is to compute 
\begin{equation}
S[T^3; (a_i); (f,g)] := - Q_{T^3,(a_i)}(\hat f\hat T_X+\hat g\hat T_Y).
\end{equation}

We first use the the quadratic refinement property to rewrite\begin{equation}
Q_{T^3,(a_i)}(\hat f\hat T_X+\hat g\hat T_Y)
= Q_{T^3,(a_i)}(\hat f\hat T_X) + Q_{T^3,(a_i)}(\hat g\hat T_Y)+ \int_{T^3} \hat f \hat T_X \hat g \hat T_Y \,.
\label{666}
\end{equation} 
Let us start by evaluating the last term.
A few basic properties of the differential cohomology integration come in handy, so let us list them.
In general, given a product manifold $M\times N$ and a differential cohomology class 
$\hat\omega \in \hat H^p(M\times N)$, we can integrate along $M$, which reduces
the degree by $\dim M$, so we have \begin{equation}
\int_M \hat \omega \in \hat H^{p-\dim M}(N).
\end{equation}
In particular, when $N$ is a point, the integral takes the values in $\hat H^*(pt)$,
which satisfies \begin{equation}
\hat H^d(pt)=\begin{cases}
\bZ & (d=0),\\
\bR/\bZ & (d=1),\\
0 & \text{otherwise}.
\end{cases}
\end{equation}
Finally, given $\hat a\in \hat H^p(M)$ and $\hat b\in \hat H^q(N)$, we have \begin{equation}
\int_{M\times N} \hat a \hat b = \int_M \hat a \int_N \hat b.
\end{equation}

Let us come back to the computations of \eqref{666}. 
The last term can be manipulated  \begin{equation}
\int_{M_3} \hat f \hat T_X \hat g \hat T_Y
=-\int_{S^1_A}\hat f \hat g
\int_{S^1_X} \hat T_X 
\int_{S^1_Y} \hat T_Y
=-\int_{S^1_A}\hat f \hat g=-\tilde \com( f, g),
\end{equation}
using the properties of differential cohomology integration listed above.
The notation $\tilde{\com}$ was introduced in \eqref{eq:gamma-tilde}.
What remains to be done is, then, to evaluate $Q_{T^3,(a_i)}(\hat f\hat T_X)$. 
The computation of $Q_{T^3,(a_i)}(\hat g\hat T_Y)$ is entirely analogous.

\if0
\comment{YT: This paragraph is dubious.}
To do this, we recall  our alternative definition \eqref{eq:QQ} of $Q_{M_3,\sigma}(\hat A)$.
We regard $M_3=T^3=S^1_X\times S^1_Y \times S^1_A$
as $(S^1_A\times S^1_X)\times S^1_Y$,
and use $N_4$ by filling in $S^1_Y$ via a disk $D^2_Y$,
such that $\partial D^2_Y= S^1_Y$.
Then \begin{equation}
Q_{T^3,(a_i)}(\hat f\hat T_X)
= \int_{S^1_A\times S^1_X} \hat f  \hat T_X  \int_{D^2_Y} \frac{F_Y}{2\pi}.
\end{equation}
Here, the $U(1)$ connection on $D^2_Y$ is there
to produce the spin structure on its boundary $S^1_Y$,
and therefore \begin{equation}
 \int_{D^2_Y} \frac{F_Y}{2\pi} = \begin{cases}
 0 & (a_Y=0),\\
 \frac12 & (a_Y=1).
\end{cases}
\label{eq:01}
\end{equation}
The first factor can be evaluated as \begin{equation}
\int_{S^1_A\times S^1_X} \hat f  \hat T_X  
=\int_{S^1_A\times S^1_X} df\frac{dx}{2\pi} = w_f.
\end{equation}
We conclude that \begin{equation}
Q_{T^3}(\hat f\hat T_X)=\frac12 a_Y w_f,
\label{eq:QfT}
\end{equation}
where we remind the reader that $a_Y=0$ or $1$ depending on whether $S^1_Y$ is NS or R.
\comment{YT: to here.}
\fi

We do this by using the relation \eqref{eq:to-eta}.
%This uses another definition of $Q_{M_3}(\hat A)$ 
%as the $\eta$-invariant of the $U(1)$ connection $\hat A$ on the spin 3-manifold $M_3$.
We  regard our $M_3=T^3$ as $(S^1_A\times S^1_X)\times S^1_Y$.
The gauge field configuration $\hat f \hat T_X$ is defined on $S^1_A\times S^1_X$
and pulled back trivially along $S^1_Y$.
Therefore we can use the product formula of the $\eta$-invariant,
\begin{equation}
(\text{$\eta$ on $M_{2n}\times S^1$}) = 
(\text{index on $M_{2n}$})
\times (\text{$\eta$ on $S^1$}) \,.
\label{eq:eta-product-formula}
\end{equation}
In our case, the first factor is simply the integral of the $U(1)$ field strength,
and is given by  \begin{equation}
\int_{T^2} \hat f \hat T_X=w_f.
\end{equation}
The second factor is $0$ or $1/2$ depending on whether the spin structure around $S^1$ 
is NS or R, respectively.
We conclude that \begin{equation}
Q_{T^3}(\hat f\hat T_X)=\frac12 a_Y w_f,
\label{eq:QfT}
\end{equation}
where we remind the reader that $a_Y=0$ or $1$ depending on whether $S^1_Y$ is NS or R.
%This second derivation generalizes to the case of four-dimensional $SU(2)$ symmetry more readily.

We therefore have found \begin{equation}
S[T^3;(a_i); f,g] = -\frac12 (w_f a_Y + a_X w_g)+\tilde \com(f,g).
\label{eq:3d-action-result}
\end{equation}
We now compare this with \eqref{get-nu} and \eqref{get-c}
to read off the fermion parity $\mu(f)$ and the commutator map $\com(f,g)$.
Taking $g=e$, the constant map sending $S^1_Y$ to the identity,
and comparing \eqref{eq:3d-action-result} against \eqref{get-nu}, we find
%\footnote{%
%Comparing the values of the invertible phase $Z = e^{2\pi i S}$ yields $(-1)^{\mu(f)} = (-1)^{w_f}$. Although $\mu(f)$ is $\bZ_2$-valued, in this context it is harmless to represent it by the integer $w_f$.} 
\begin{equation}
\mu(f)= w_f \mod 2. \label{eq:nu-w}
\end{equation}
Taking both $S^1_X$ and $S^1_Y$ to be NS, 
we have $a_X=a_Y=0$. 
Comparing \eqref{eq:3d-action-result} against \eqref{get-c} and using \eqref{eq:nu-w}, we have \begin{equation}
\com(f,g)=\tilde \com(f,g) - \frac12 w_fw_g.
\end{equation}
This is the relation we already pointed out in \eqref{eq:tilde-c-vs-c}.

\subsection{Position-dependent \texorpdfstring{$SU(2)$}{SU(2)} transformation in four dimensions} 
\label{sec:YYY}

\subsubsection{The property to be derived}
The analysis above can be generalized to $SU(2)$ symmetry in four dimensions.
Namely, we want to consider four-dimensional theory with $SU(2)$ symmetry on 
$M_3$ times the time direction,
and consider the action of position-dependent symmetry operation
specified by $M_3\to SU(2)$.
Take $f,g : M_3\to SU(2)$ whose supports are disjoint balls. 
Then $f$ and $g$ commute. But how about $U(f)$ and $U(g)$,
the unitary operators which implement position-dependent $SU(2)$ symmetry operations
on the Hilbert space $\cH$ of the theory?
We will find that \begin{equation}
U(f)U(g)=U(g)U(f)
\end{equation} when the $SU(2)$ symmetry is non-anomalous,
but we have \begin{equation}
U(f)U(g)=(-1)^{w_f w_g} U(g)U(f)\label{eq:4dfggf}
\end{equation} when the $SU(2)$ symmetry has the Witten anomaly \cite{Witten:1982fp},
where $w_{f,g}$ are the winding numbers of the maps $f, g:M_3\to SU(2)$.
This, for example, gives a Hamiltonian understanding of why it is impossible to gauge $SU(2)$
if there is the Witten anomaly, as there simply is no state in the Hilbert system which is invariant under all $U(f)$.

Before going into the computation,
we should mention that this result \eqref{eq:4dfggf} is not very surprising
if we use the following heuristic argument, employing the viewpoint B of footnote \ref{foot:fundamental}.
We start from a 4d $SU(n)$ QCD with two flavors.
It has $SU(2)_L\times SU(2)_R$ flavor symmetry.
This is expected to confine and to be described by an $SU(2)$ sigma model in the infrared,
where $SU(2)_L$ and $SU(2)_R$ act from the left and the right of this sigma model manifold.

Let us focus on $SU(2)_L$. This has Witten anomaly if and only if $n$ is odd.
We now consider the operator $U(f)$ for a map $f:\bR^3\to SU(2)$ of winding number $1$.
When $U(f)$ is acted upon to the vacuum in the infrared description,
this creates a Skyrmion, as $f$ acts on the $SU(2)$ sigma model field 
as a chiral $SU(2)$ action.
A Skyrmion is a low-energy sigma model representation of a baryon,
and is therefore a boson or a fermion depending on whether $n$ 
is even or odd \cite{Witten:1983tw,Witten:1983tx}.
This means that, when $n$ is odd, $U(f)$ is a creation operator of a fermionic Skyrmion.
Therefore, such operators $U(f)$ and $U(g)$ should anticommute
when the support of $f$ and the support of $g$ are disjoint.
Finally, the commutation relation of $U(f)$ and $U(g)$ is a topological property of the system
determined by the anomaly of the symmetry in question,
and in particular should not depend on the choice of the concrete theory discussed.
This should mean that the anticommutation thus found should be a universal
consequence of the Witten anomaly.
What we will provide below is the computation of this commutation relation
using the formalism we developed in this paper,
in the viewpoint A of footnote \ref{foot:fundamental}.

\subsubsection{Computation}

To compute the commutation relation,
we consider $M_5=M_3\times S^1_X\times S^1_Y$,
where $M_3$ plays the role of $S^1_A$ in the two-dimensional case.
We then construct an $SU(2)$ bundle on $M_3\times S^1_X$ using the gauge transformation $f$,
and an $SU(2)$ bundle on $M_3\times S^1_Y$ using the gauge transformation $g$.\footnote{%
An explicit gauge field configuration on $M_3\times S^1_Y$ can be chosen as follows.
We start from a trivial $SU(2)$ bundle over $M_3 \times [0,1]$, 
and parameterize $[0,1]$ by $t$. 
The gauge field is given by an adjoint valued 1-form $A(t)$, where the dependence on 
the coordinates on $M_3$ is left implicit in the notation.
We need $A(1)=g^{-1} A(0) g + g^{-1} d g$. 
This can be achieved by setting $A(t)=h(t) g^{-1} dg$, 
where $h$ is a monotonic function such that $h(t)=0$ for $0\le t<a$,
$h(t)=1$ for $b<t\le 1$ and interpolating between $0$ and $1$ in $a<t<b$.
By construction, this configuration has instanton number one,
and the curvature is nonzero only in the region where $g$ is nontrivial
and $a<t<b$ at the same time.
}
The five-dimensional configuration is as drawn in Fig.~\ref{fig:terrible-drawing}.
We would like to evaluate the bulk invertible phase in this configuration,
which we denote as \begin{equation}
S[M_5,\sigma,(f,g)] \in \bR/\bZ.
\end{equation}

\begin{figure}[h]
\centering
\includegraphics[width=1.15\textwidth]{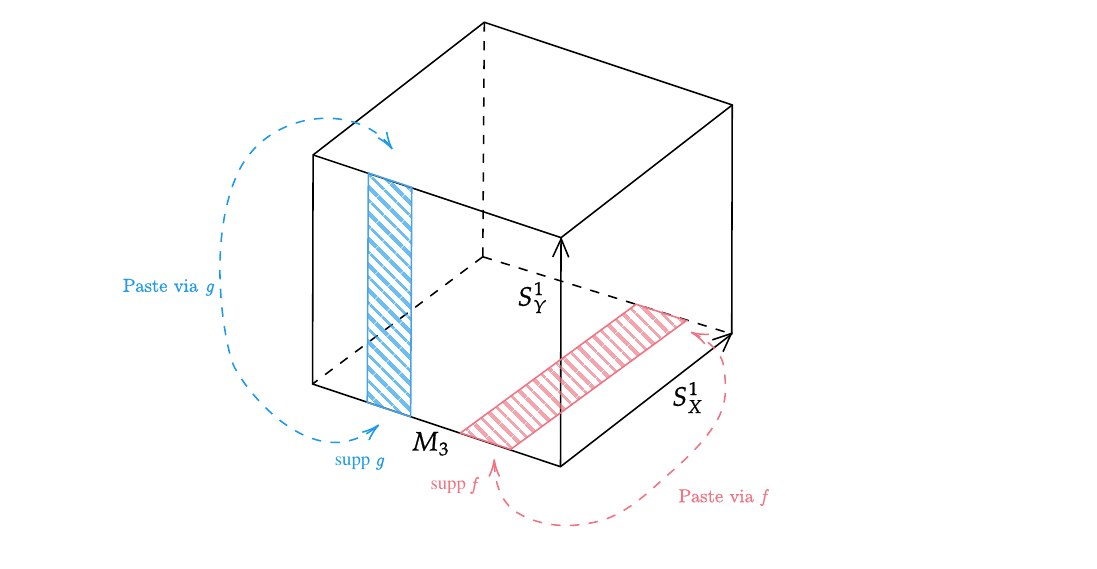}
\caption{The 5d configuration we employ.
The light-red shaded region is where the $SU(2)$ field is nontrivial due to the gauge transformation by $f$.
(This entire gauge configuration on $M_3\times S^1_X$ is pulled back rather trivially to $S^1_Y$,
and therefore it might be better to fill by light-red along the $S^1_Y$ direction.
That would make the figure more complicated, so we decided not to do so.)
Similarly, the light-blue shaded region is from $g$.
\label{fig:terrible-drawing}}
\end{figure}

The five-dimensional bulk invertible phase is a quantum field theory,
and therefore has the pasting property.
Namely, decompose $M_3= (\supp f ) \sqcup (\supp g) \sqcup M_3'$,
and pick two functions $\tilde f$, $\tilde g$ not necessarily equal to $f$, $g$
but such that $\supp \tilde f\subset \supp f$ and $\supp \tilde g\subset \supp g$.
Then we have\footnote{%
Note that three terms on the right hand side of this equation has the following subtlety.
The exponential of the second term, for example,
should be thought of as a unitary map $$
e^{2\pi iS[M_3'\times T^2_{X,Y}, \sigma, (e,e)]}
: \cH( \partial(\supp f)\times T^2_{X,Y})
\to \cH( \partial(\supp g)\times T^2_{X,Y})
$$
where $\cH(M_4)$ is the Hilbert space of the invertible theory
on the manifold $M_4$ (which plays the role of the constant time slice) with trivial gauge field.
$\cH(M_4)$ for any $M_4$ are all one dimensional, 
but there is no canonical identification with $\mathbb{C}$.
Therefore $e^{2\pi iS[M_3'\times T^2_{X,Y}, \sigma, (e,e)]}$
is a number only after choosing the bases of $\cH( \partial(\supp f)\times T^2_{X,Y})$
and  $\cH( \partial(\supp g)\times T^2_{X,Y})$.
This dependence on the choice of the bases is a manifestation of the anomaly of the  boundary theory of this 
invertible field theory.
An important point for us is that these anomalous variations cancel out when we sum these three terms,
and the left hand side of \eqref{S=SSS} is naturally a number.
To carry out the argument which follows, we need to pick the bases of $\cH(M_4)$ 
once for each $M_4$ appearing in the discussion, and keep them fixed throughout.
}
\begin{multline}
S[M_5,\sigma,(\tilde f,\tilde g)] = 
S[(\supp f )\times T^2_{X,Y}, \sigma, (\tilde f,e)]\\
+S[M_3' \times T^2_{X,Y}, \sigma, (e,e)]
+S[(\supp g )\times T^2_{X,Y}, \sigma, (e, \tilde g)].
\label{S=SSS}
\end{multline}
From this we can check that the following identity 
\begin{equation}
S[M_5,\sigma,(f,g)] 
+ S[M_5,\sigma,(e,e)] 
= S[M_5,\sigma,(f,e)] + S[M_5,\sigma,(e,g)]
\label{SS=SS}
\end{equation}
is satisfied, by simply rewriting both sides of \eqref{SS=SS} using \eqref{S=SSS}.
This equation allows us to determine $S[M_5,\sigma,(f,g)]$
from $S[M_5,\sigma,(f,e)]$, $S[M_5,\sigma,(e,g)]$ and $S[M_5,\sigma,(e,e)]$,
which are simpler to compute.

The remaining task then is to evaluate $S[M_5,\sigma,(f,e))]$;
the other two are entirely analogous.
The bulk manifold is now $M_5 = (M_3\times S^1_X)\times S^1_Y$,
such that the holonomy of the $SU(2)$ transformation around $S^1_X$ is 
specified by $f:M_3\to SU(2)$, 
while the holonomy $S^1_Y$ is taken to be the identity $e$.
The whole configuration is therefore the product
of an $SU(2)$ connection defined on $M_4=M_3\times S^1_X$
and a configuration on $S^1_Y$ which is trivial except for the spin structure.
The construction up to this point was completely general.
Now we specialize to the case of the Witten anomaly,
for which the invertible phase is the $\eta$-invariant for the fermion in the doublet representation
of the $SU(2)$ gauge field.
We can then use  the product formula \eqref{eq:eta-product-formula} to finish  the computation.
We have \begin{align}
S[M_5,\sigma,(f,e))]
= \eta(M_4\times S^1_Y)
= (\text{index on $M_4$}) \times (\text{$\eta$ on $S^1_Y$}).
\label{eq:M4prod}
\end{align}

The $SU(2)$ connection on $M_4=M_3 \times S^1_X$
is constructed by taking a trivial configuration on $M_3\times [0,1]$
and gluing the two ends of $[0,1]$ with a gauge transformation $f:M_3\to SU(2)$.
This means that the index of the $SU(2)$ connection 
on $M_4$ is the winding number $w_f$ of $f$.
The $\eta$-invariant on $S^1_Y$ is $a_Y/2$,
where $a_Y$ is $0$ or $1$ when $S^1_Y$ is in the NS or R-sector, respectively.
Plugging these information in \eqref{eq:M4prod}, we find 
\begin{equation}
S[M_5,\sigma,(f,e))] =\frac12 a_Y w_f \mod 1.
\end{equation}
From this, we find the fermion number $\mu(f)$ of the unitary operator $U(f)$ to be given by $\mu(f)=w_f$ as before.
Then taking $a_X=a_Y=0$, we see $S[M_5,\sigma,(f,g)]=0$.
Comparing with \eqref{get-c} and using $\mu(f)=w_f$, we find \begin{equation}
\com(f,g)=  \frac{1}{2}w_f w_g \mod 1
\end{equation}
exactly as before.

\section*{Acknowledgements}

The authors thank discussions with Kotaro Kawasumi.
YT thanks discussions with Sahand Seifnashri and Nikita Sopenko back in the spring of 2024 in Princeton
which became a source of inspiration of this paper.
SS  thanks Kantaro Ohmori for the discussions on the content of Sec.~\ref{sec:1st-derivation}.

MO, YT, and YZ are supported in part  
by WPI Initiative, MEXT, Japan at Kavli IPMU, the University of Tokyo.
MO is supported by FoPM (WINGS Program of the University of Tokyo), 
JSPS Research Fellowship for Young Scientists, and
JSPS KAKENHI Grant No.\ JP23KJ0650.
SS is supported by the World-leading Innovative Graduate Study Program for Frontiers of Mathematical Sciences and Physics (WINGS-FMSP), the University of Tokyo.
YT is supported by JSPS KAKENHI Grant-in-Aid (Kiban-C), No.24K06883. 
YZ is supported by National Science Foundation of China (NSFC) under Grant No. 12305077. 

\appendix

%%%%%%%%%%%%%%%%%%%%%%%%%%%%%%
\section{Explicit calculation of the commutator map with  \texorpdfstring{$n$}{n} patches}
\label{sec:general-n-calc}
%%%%%%%%%%%%%%%%%%%%%%%%%%%%%%

In this appendix, we explicitly calculate the commutator map of the position-dependent symmetry operator $U_\text{no}(f)$ with general $n$ patches:
\begin{align}
\label{eq:trun_sym_op_ap}
    U_\text{no}(f) ={} : \exp \left[ 2\pi i \left( \sum_{u=1}^n \int_{\sigma_u} dx\, f_u(x)  J_0 (x) - \sum_{u = 1}^n w^f_{u,u+1}\, \chi(\sigma_{u,u+1})  \right) \right] : \,.
\end{align}
From the mode expansion \eqref{Jexp} and \eqref{chiexp}, $U_\text{no}(f)$ is written as
\begin{align}
    U_{\text{no}}(f)  =  e^{2\pi i A_{-}(f)}  e^{2\pi i B_{0}(f)} e^{2\pi i A_{0}(f)}   e^{2\pi i A_{+}(f)} \,,
\end{align}
where
\begin{align}
    A_{-}(f) &= \sum_{s=1}^{\infty} \left( \sum_{u=1}^n \int_{\sigma_u} dx\, f_u(x) \frac{1}{2\pi} e^{-isx} - \sum_{u = 1}^n w^f_{u,u+1}\, \frac{1}{2\pi i(-s)} e^{-is\sigma_{u,u+1}}  \right)  J_{-s} \,, \\
    B_{0}(f) &= \left( - \sum_{u = 1}^n w^f_{u,u+1}\, \frac{1}{2\pi}  \right) I_0 \,, \\
    A_{0}(f) &= \left( \sum_{u=1}^n \int_{\sigma_u} dx\, f_u(x) \frac{1}{2\pi} - \sum_{u = 1}^n w^f_{u,u+1}\, \frac{1}{2\pi} \sigma_{u,u+1} \right) J_0 \,, \\
    A_{+}(f) &= \left( \sum_{u=1}^n \int_{\sigma_u} dx\, f_u(x) \frac{1}{2\pi} e^{isx} - \sum_{u = 1}^n w^f_{u,u+1}\, \frac{1}{2\pi is} e^{is\sigma_{u,u+1}}  \right) J_s \,. 
\end{align}
Considering the product between $U_\text{no}(f)$ and $U_\text{no}(g)$, we can always take the same patches $\{\sigma_u \mid u = 1,2,...,n\}$ of $f$, $g$ and $f+g$. From the commutation relations, we can write $U_{\text{no}}(f+g)$ as
\begin{align}
\begin{split}
    U_{\text{no}}(f+g) &= e^{2\pi i A_{-}(f+g)}  e^{2\pi i B_{0}(f+g)} 
 e^{2\pi i A_{0}(f+g)}   e^{2\pi i A_{+}(f+g)} \\
    &= e^{2\pi i A_{-}(f)} e^{2\pi i A_{-}(g)} e^{2\pi i B_{0}(f)} e^{2\pi i B_{0}(g)}  e^{2\pi i A_{0}(f)}  e^{2\pi i A_{0}(g)}  e^{2\pi i A_{+}(f)} e^{2\pi i A_{+}(g)} \,.
\end{split}
\end{align}
The commutators $[A_+(f), A_-(g)]$ and $[A_0(f), B_0(g)]$ are $c$-numbers, so we have
\begin{align}
\begin{split}
    U_{\text{no}}(f) \cdot U_{\text{no}}(g) &= e^{2\pi i A_{-}(f)} e^{2\pi i B_{0}(f)}  e^{2\pi i A_{0}(f)}   e^{2\pi i A_{+}(f)} \cdot e^{2\pi i A_{-}(g)} e^{2\pi i B_{0}(g)} e^{2\pi i A_{0}(g)}   e^{2\pi i A_{+}(g)}  \\
    &= e^{2\pi i \cdot 2\pi i \left(  [A_+(f), A_-(g)] + [A_0(f), B_0(g)] \right) }\  U_{\text{no}} (f+g) \,.
\end{split}
\end{align}
Then, the 2-cocycle for $n$ patches is given by
\begin{align}
\begin{split}
\label{eq:2-cocycle_div}
    \coc_n(f,g) &=  2 \pi i \left( [A_+(f), A_-(g)] + [A_0(f), B_0(g)] \right) \\
    &= k \left[ \left( \sum_{u=1}^n \int_{\sigma_u} dx\, \sum_{v=1}^n \int_{\sigma_v} dy\, f_u(x) \partial_y g_v(y) \right) \sum_{s=1}^{\infty} \frac{1}{2\pi} e^{is(x - y)} \right. \\
    & \quad -   \left( \sum_{u=1}^n \sum_{v=1}^n \int_{\sigma_v} dy\, w^f_{u,u+1} g_v(y) \right)\sum_{s=1}^{\infty}  \frac{1}{2\pi} e^{is(\sigma_{u,u+1} - y)}  \\
    & \quad +  \left( \sum_{u=1}^n \sum_{v=1}^n  w^f_{u,u+1} w^g_{v,v+1} \right) \sum_{s=1}^{\infty} \frac{i}{2\pi s} e^{is(\sigma_{u,u+1} - \sigma_{v,v+1})} \\
    & \left. \quad +  \left( \sum_{u=1}^n \sum_{v=1}^n  \int_{\sigma_u} dx\, f_u (x) w^g_{v,v+1} \right) \frac{1}{2\pi} -  \left( \sum_{u=1}^n \sum_{v=1}^n   w^f_{u,u+1} w^g_{v,v+1} \right) \frac{1}{2\pi} \sigma_{u,u+1} \right]\,.
\end{split}
\end{align}
We can also have another expression by integrating by parts:
\begin{align}
\begin{split}
    \coc_n(f,g) 
    &= - k \left[ \left( \sum_{u=1}^n \int_{\sigma_u} dx\, \sum_{v=1}^n \int_{\sigma_v} dy\, \partial_x f_u(x)\,  g_v(y) \right) \sum_{s=1}^{\infty} \frac{1}{2\pi} e^{is(x - y)} \right. \\
    & \quad +   \left( \sum_{u=1}^n \int_{\sigma_u} dx\,  \sum_{v=1}^n  f_u(x) w^g_{v,v+1}  \right)\sum_{s=1}^{\infty}  \frac{1}{2\pi} e^{is(x - \sigma_{v,v+1})} \\
    & \quad + \left( \sum_{u=1}^n \sum_{v=1}^n  w^f_{u,u+1} w^g_{v,v+1} \right) \sum_{s=1}^{\infty} \frac{i}{2\pi s} e^{is(\sigma_{u,u+1} - \sigma_{v,v+1})} \\
    & \left. \quad + \left( \sum_{u=1}^n \sum_{v=1}^n  \int_{\sigma_u} dx\, f_u (x) w^g_{v,v+1} \right) \frac{1}{2\pi} - \left( \sum_{u=1}^n \sum_{v=1}^n   w^f_{u,u+1} w^g_{v,v+1} \right) \frac{1}{2\pi} \sigma_{u,u+1} \right] \,.
\end{split}
\end{align}
These expressions $\coc_n(f,g)$ are divergent, so we regularize \eqref{eq:2-cocycle_div} as
\begin{align}
\begin{split}
\label{eq:2-cocycle_reg}
    \coc_{n,\epsilon}(f,g) &= k\left[ \left( \sum_{u=1}^n \int_{\sigma_u} dx\, \sum_{v=1}^n \int_{\sigma_v} dy\, f_u(x) \partial_y g_v(y) \right) \sum_{s=1}^{\infty} \frac{1}{2\pi} e^{is(x - y) - s \epsilon} \right.  \\
    & \quad -   \left( \sum_{u=1}^n \sum_{v=1}^n \int_{\sigma_v} dy\, w^f_{u,u+1} g_v(y) \right)\sum_{s=1}^{\infty}  \frac{1}{2\pi} e^{is(\sigma_{u,u+1} - y) - s \epsilon} \\
    & \quad +  \left( \sum_{u=1}^n \sum_{v=1}^n  w^f_{u,u+1} w^g_{v,v+1} \right) \sum_{s=1}^{\infty} \frac{i}{2\pi s} e^{is(\sigma_{u,u+1} - \sigma_{v,v+1}) - s \epsilon} \\
    & \left. \quad + \left( \sum_{u=1}^n \sum_{v=1}^n  \int_{\sigma_u} dx\, f_u (x) w^g_{v,v+1} \right) \frac{1}{2\pi} - \left( \sum_{u=1}^n \sum_{v=1}^n   w^f_{u,u+1} w^g_{v,v+1} \right) \frac{1}{2\pi} \sigma_{u,u+1} \right] \,,
\end{split}
\end{align}
where $\epsilon$ is an infinitesimal positive real number.\footnote{
For the free massless theory, the mode of $J^t(t,x)$ is given by $J_n(t) = J_n e^{int}$. Then, this regularization can be obtained by considering a product $U_{\text{no}}(f)[t+i\epsilon/2] \cdot U_{\text{no}}(g)[t-i\epsilon/2]$ where $U_{\text{no}}(f)$ and $U_{\text{no}}(g)$ are supposed to be at slightly different times. 
Other slightly different regularizations can be obtained by 
applying the same argument where the mode is written as $J_n(t) = J_n e^{i\omega(n)t}$ and
$\omega(n)$ satisfies a suitable dispersion relation.
We  consider the regularized 2-cocycle \eqref{eq:2-cocycle_reg} simply for brevity.
}
The commutator map with $\epsilon$ is given by
\begin{align}
\begin{split}
    \com_{n,\epsilon}(f,g) &= \coc_{n,\epsilon}(f,g) - \coc_{n,\epsilon}(g,f)  \\
    &= k \left[ \left( \sum_{u=1}^n \int_{\sigma_u} dx\, \sum_{v=1}^n \int_{\sigma_v} dy\, f_u(x) \partial_y g_v(y) \right) \delta_\epsilon^{(\text{P})}(x - y) \right. \\
    & \quad -    \left( \sum_{u=1}^n \sum_{v=1}^n \int_{\sigma_v} dy\, w^f_{u,u+1} g_v(y) \right) \delta_\epsilon^{(\text{P})}(\sigma_{u,u+1} - y) \\
    &\left. \quad -  \left( \sum_{u=1}^n \sum_{v=1}^n  w^f_{u,u+1} w^g_{v,v+1} \right) \theta_\epsilon^{(\text{P}')} (\sigma_{u,u+1} - \sigma_{v,v+1}) \right] \,,
\end{split}
\end{align}
where
\begin{align}
    \delta_\epsilon^{(\text{P})} (x) &= \sum_{n \in \mathbb{Z}} \frac{1}{2\pi} e^{-inx - |n| \epsilon} \,,\\
    \theta_\epsilon^{(\text{P}')} (x) &= \sum_{n \in \mathbb{Z} \setminus \{0\} } \frac{i}{2\pi n} e^{-inx - |n| \epsilon} + \frac{1}{2\pi} x \,.
\end{align}
In this formula, we can take the limit $\epsilon \rightarrow +0$, such that we have $\delta_\epsilon^{(\text{P})} (x) \rightarrow \delta^{(\text{P})} (x) $ and $\theta_\epsilon^{(\text{P}')} (x) \rightarrow\theta^{(\text{P}')} (x)$, and then the 2-commutator map becomes
\begin{align}
\begin{split}
    \com_n(f,g) &= \lim_{\epsilon \rightarrow +0} \com_{n,\epsilon}(f,g) \\
    &= k \left[ \left( \sum_{u=1}^n \int_{\sigma_u} dx\, \sum_{v=1}^n \int_{\sigma_v} dy\, f_u(x) \partial_y g_v(y) \right) \delta^{(\text{P})}(x - y) \right. \\
    & \quad -    \left( \sum_{u=1}^n \sum_{v=1}^n \int_{\sigma_v} dy\, w^f_{u,u+1} g_v(y) \right) \delta^{(\text{P})}(\sigma_{u,u+1} - y) \\
    &\left. \quad -  \left( \sum_{u=1}^n \sum_{v=1}^n  w^f_{u,u+1} w^g_{v,v+1} \right) \theta^{(\text{P}')} (\sigma_{u,u+1} - \sigma_{v,v+1}) \right] \,.
\end{split}
\end{align}
Note that
\begin{align}
\begin{split}
    \delta_\epsilon^{(\text{P})} (x) &= \sum_{n \in \mathbb{Z}} \frac{1}{2\pi} e^{-inx - |n| \epsilon} \\
    &= \frac{1}{2\pi} \left(1 +  \frac{1}{e^{ix + \epsilon} - 1} + \frac{1}{e^{-ix + \epsilon} - 1} \right) 
\end{split}
\end{align}
is an even function. Therefore, we use the natural endpoint regularization of the periodic delta function
\begin{align}
    \int^{b}_{a}dx\, f(x) \delta^{(\text{P})}(x-a) = \frac{1}{2}f(a) \,,\\
    \int^{b}_{a}dx\, f(x) \delta^{(\text{P})}(x - b) = \frac{1}{2}f(b) \,.
\end{align}
Performing the integral of the periodic delta function, we obtain
\begin{align}
\begin{split}
    \com_n(f,g) 
    &= k \left[  \sum_{u=1}^n \int_{\sigma_u} dx\,  f_u(x) g'_u(x)  \right. \\
    & \quad -   \sum_{u=1}^n  w^f_{u,u+1}  \left( \frac{1}{2} g_u(\sigma_{u,u+1}) + \frac{1}{2} g_{u+1}(\sigma_{u,u+1}) \right)  \\
    &\left. \quad -  \left( \sum_{u=1}^n \sum_{v=1}^n  w^f_{u,u+1} w^g_{v,v+1} \right) \theta^{(\text{P}')} (\sigma_{u,u+1} - \sigma_{v,v+1}) \right]  \,.
\end{split}
\end{align}
Integration by parts yields the expression
\begin{align}
\begin{split}
\label{eq:commap_n}
    \com_n(f,g) 
    = &\frac{k}{2} \left[ \sum_{u=1}^n \int_{\sigma_u} dx\, \left( f_u(x) g'_u(x) - f'_u(x) \, g_u(x) \right) \right. \\
    &   \quad + \sum_{u=1}^n  f_u(\sigma_{u,u+1}) w^g_{u,u+1} - \sum_{u=1}^n w^f_{u,u+1}  g_u(\sigma_{u,u+1}) \\
    & \quad \left. + \sum_{1 \leq u < v \leq n} ( w^f_{u,u+1} w^g_{v,v+1} - w^f_{v,v+1} w^g_{u,u+1} )  \right] \,.
\end{split}
\end{align}
Note that the commutator map $\com_n(f,g)$ is defined modulo $1$, so we can freely add integer terms.
For $n=1$ (see Fig.~\ref{fig:patch1}), we have
\begin{align}
\begin{split}
    \com_1(f,g) 
    &= \frac{k}{2} \left[ \int_{\sigma_1} dx\,  \left( f_1(x) g'_1(x) - f'_1(x) \, g_1(x) \right) \right. \\
    & \left. \quad +  f_1(\sigma_{1,2}) w^g_{1,2} -  w^f_{1,2}  g_1(\sigma_{1,2})  \right] \\
    &= \com(f,g) \,,
\end{split}
\end{align}
where $\com(f,g)$ is the one that appeared in \eqref{eq:commap_1}.

We can also consider the case where $f$ and $g$ are each defined on a different single patch and they jump at different  points. 
We let the winding points of $f$ and $g$ be $0$ and $p$ respectively. Then, instead of considering that case, we may consider the case $n = 2$ where the winding points are $[\sigma_{0,1} = 0] \sim [\sigma_{2,3} = 2\pi]$ and $\sigma_{1,2} = p$, and the winding integers are $w^f_{2,1} = w^f_{0,1} = w^f_{2,3}= w_f$, $w^f_{1,2} = 0$, $w^g_{2,1} = w^g_{0,1} = w^g_{2,3} = 0$ and $w^g_{1,2} = w_g$ (see Fig.~\ref{fig:patch2}). The commutator map is
\begin{align}
%\begin{split}
    \com_2(f,g) 
    &= \frac{k}{2} \biggl[ \int_{\sigma_1} dx\, \left( f_1(x) g'_1(x) - f'_1(x) \, g_1(x) \right) + \int_{\sigma_2} dx\, \left( f_2(x) g'_2(x) - f'_2(x) \, g_2(x) \right) \nonumber\\
    & \quad +  f_1(\sigma_{1,2}) w_g  -  w_f  g_2(\sigma_{2,3})  - w^f_{2,3} w^g_{1,2} \biggr] \\
    &= \frac{k}{2} \left[ \int_0^{2\pi} dx\, \left( f(x) g'(x) - f'(x) \, g(x) \right)     + f(p) w_g  -  w_f  g(0)  - w_f w_g \right] \,,
%\end{split}
\end{align}
where $f(x)$ and $g(x)$ are
\begin{align}
    f(x) = 
\begin{cases}
    f_1(x) \,, & 0 \leq x < p  \,, \\
    f_2(x) \,, & p \leq x \leq 2\pi \,,
\end{cases}
\end{align}
and
\begin{align}
    g(x) = 
\begin{cases}
    g_1(x) \,, & 0 \leq x < p  \,, \\
    g_2(x) \,, & p \leq x \leq 2\pi \,.
\end{cases}
\end{align}

\begin{figure}
  \begin{center}
  \begin{minipage}{0.45\hsize}
%  \vspace{-1cm}
   \begin{center}
    \includegraphics[width=\hsize]{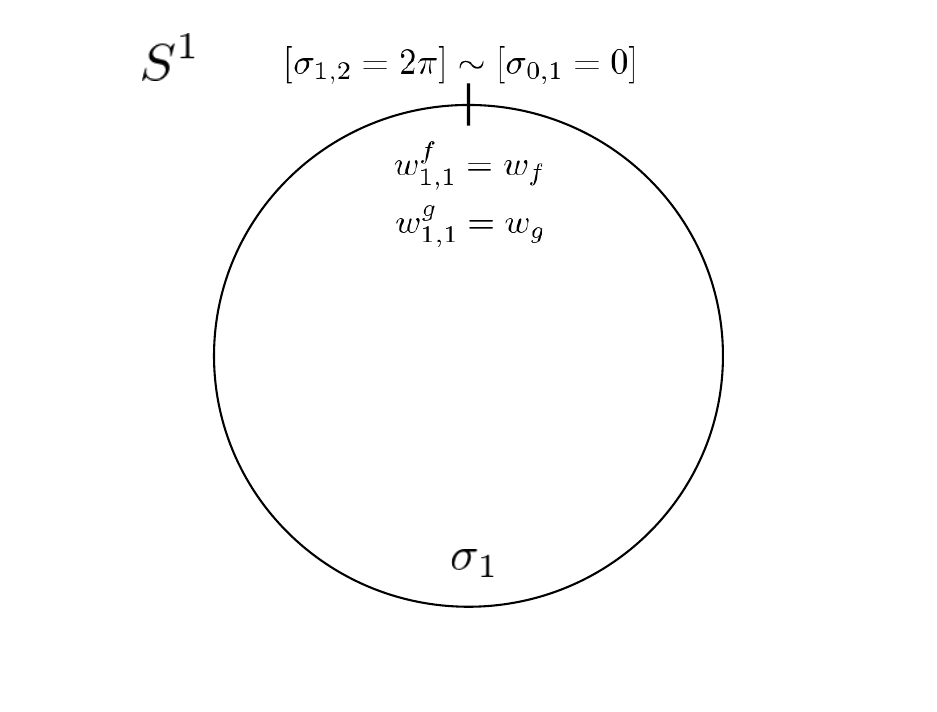} 
    \end{center}
          \vspace{-0.5cm}
      \caption{The $n=1$ patching. {\color{white} where the winding points of $f$ and $g$ are $0$ and $0$ respectively.} }
      \label{fig:patch1}
     \end{minipage}
     \hspace{0.05\hsize}
     \begin{minipage}{0.45\hsize}
      \begin{center}
       \includegraphics[width=\hsize]{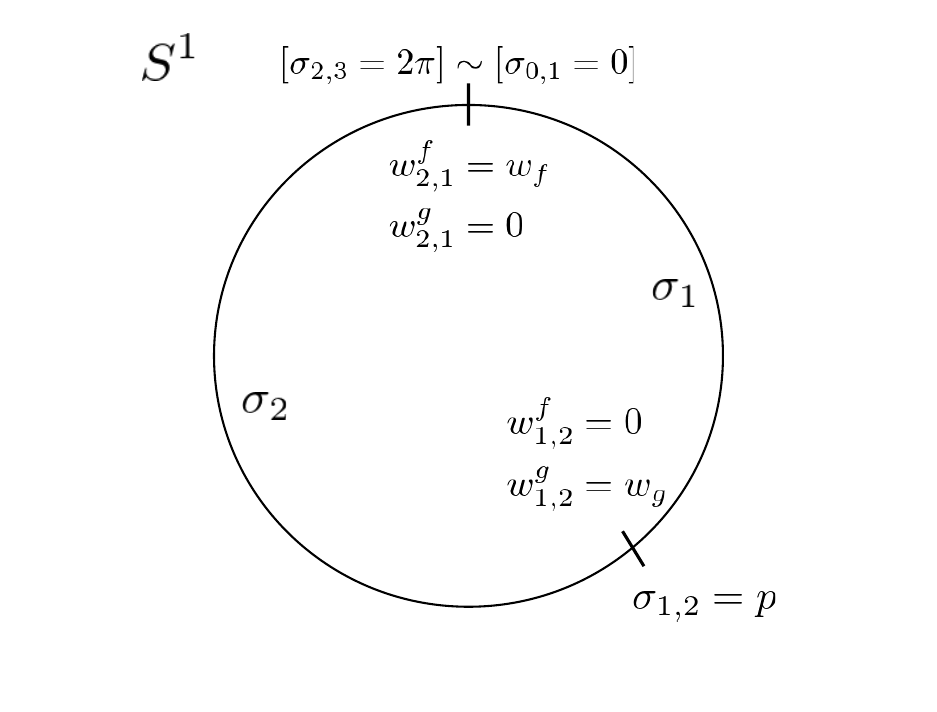}
      \end{center}
      \vspace{-0.5cm}
      \caption{The $n=2$ patching, where the winding points of $f$ and $g$ are $0$ and $p$ respectively.}
      \label{fig:patch2}
     \end{minipage}
    \end{center}
\end{figure}

%%%%%%%%%%%%%%%%%%%%%%%%%%%%%%
\section{Proof of the uniqueness of the commutator map}
\label{sec:proof-of-uniquness}
%%%%%%%%%%%%%%%%%%%%%%%%%%%%%%

In this appendix, we prove Theorem \ref{thm:uniqueness-com}, the uniqueness of the commutator map satisfying the consistency conditions ($\com$-0)--($\com$-3).

To prove this theorem,
let us prepare some notations.
We take four points $0 < a_1 < a_2 < a_3 < a_4 < 2\pi$,
and use the following two functions $z_L, z_R\in \cF$ of winding number one:
\begin{align}
z_L(x)&:=\left\{\begin{array}{ll}
0 & (0 \leq x < a_1)  \\
\text{interpolate} & (a_1 \leq x < a_2)\\
1 & (a_2 \leq x \leq 2\pi)
\end{array}\right.,
\\
z_R(x)&:=\left\{\begin{array}{ll}
0 & (0 \leq x < a_3)  \\
\text{interpolate} & (a_3 \leq x < a_4)\\
1 & (a_4 \leq x \leq 2\pi)
\end{array}\right.,
\end{align}
with $\supp z_L=[a_1,a_2]$ and $\supp z_R=[a_3,a_4]$.
We also define a function $z_0\in\cF$ of winding number zero by $z_L=z_R+z_0$.
See Fig.~\ref{fig:zzz} for an illustration.
\begin{figure}[H]
\[\vcenter{\hbox{\includegraphics[width=.9\textwidth]{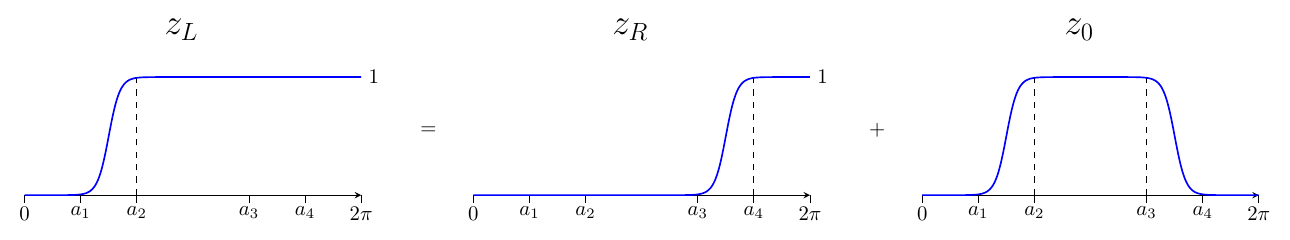}}}\]
\caption{Profile functions used in the proof.
\label{fig:zzz}}
\end{figure}

\begin{lemma}
\label{lemma:wn-0-1-commutator}
If $f_0\in\cF$ has winding number zero and satisfies $\supp f_0 \cap \supp z_L=\varnothing$,
then $\com(f_0,z_L)=0 \bmod 1$.
Similarly, the statement where $z_L$ is replaced with $z_R$ also holds.
\end{lemma}

\begin{proof}
We show the first statement.
From the locality condition, we know that $\com(f_0,z_L)\in\bR/\bZ$ is $0$ or $\frac{1}{2}$ mod $1$,
but we can show that it is in fact $0$ as follows.
By assumption, $f_0(\supp z_L)$ is an integer constant $m\in\bZ$,
and we define $\tilde{f}_0:=f_0-m$ so that $\tilde{f}_0(\supp z_L)=0$.

Since the winding number of $\tilde{f}_0$ is zero,
$\frac{1}{t}\tilde{f}_0$ is again an element of $\cF$ for arbitrary $t\in\bZ$.
Thanks to $\tilde{f}_0(\supp z_L)=0$, this $\frac{1}{t}\tilde{f}_0$ also satisfies $\supp \frac{1}{t}\tilde{f}_0 \cap \supp z_L=\varnothing$,
so we have $\com(\frac{1}{t}\tilde{f}_0,z_L) \in \frac{1}{2}\bZ$ from the locality condition.
Using the bi-additivity of $\com$, we can show $\com(\frac{1}{t}\tilde{f}_0,z_L)=\frac{1}{t}\com(\tilde{f}_0,z_L) \bmod 1$.
Now, $\frac{1}{t}\com(\tilde{f}_0,z_L) \in \frac{1}{2}\bZ$ for any $t\in\bZ$,
which means $\com(\tilde{f}_0,z_L)=0$.
Therefore, $\com(f_0,z_L)=0 \bmod 1$.
\end{proof}

Using this Lemma \ref{lemma:wn-0-1-commutator},
we can derive an explicit formula of $\com(f,g)$ for $f,g \in \cF$.
First, we define $f_0,g_0 \in \cF$ with winding number zero by $f=f_0+w_fz_L$ and $g=g_0+w_gz_L$.
By the bi-additivity and the alternating property $\com(z_L,z_L)=0$,
we have
\begin{align}
\com(f,g) = \com_0(f_0,g_0) + \com(h_0,z_L),
\end{align}
where
\begin{align}
h_0:=w_gf_0-w_fg_0 \in \cF
\end{align}
also has winding number zero.
Next, we decompose $h_0$ as
\begin{align}
h_0 = (h_0)_L + (h_0)_R,
\end{align}
where $(h_0)_L$ has support in $[0,a_3]\cup[a_4,2\pi]$ (vanishing on $[a_3,a_4]$)
and $(h_0)_R$ has support in $[0,a_1]\cup[a_2,2\pi]$ (vanishing on $[a_1,a_2]$),
by multiplying $h_0$ by a so-called partition of unity.
Since the winding number of $h_0$ is zero,
these $(h_0)_L$ and $(h_0)_R$ are again elements of $\cF$.
See Fig.~\ref{fig:decomp} for an illustration.
\begin{figure}[H]
\[\vcenter{\hbox{\includegraphics[width=.9\textwidth]{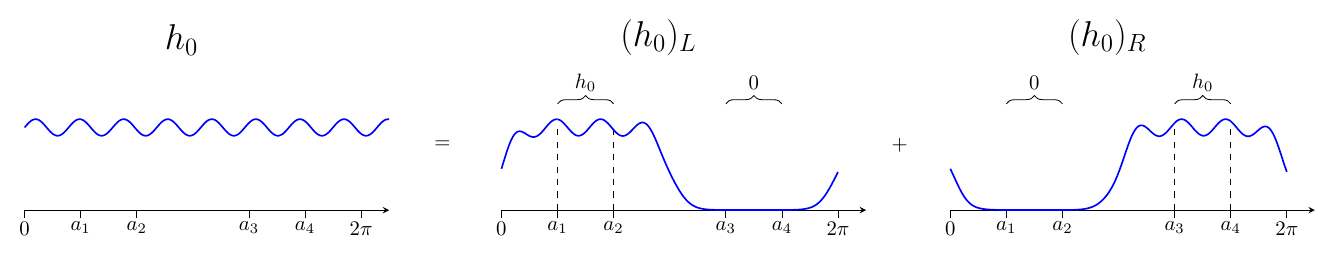}}}\]
\caption{The decomposition of $h_0$ into two parts, $h_0^L$ and $h_0^R$.
\label{fig:decomp}}
\end{figure}
Then,
\begin{align}
\com(f,g) & = \com_0(f_0,g_0) + \com((h_0)_L,z_L) + \com((h_0)_R,z_L)\\
& = \com_0(f_0,g_0) + \com_0((h_0)_L,z_0) + \com((h_0)_L,z_R) + \com((h_0)_R,z_L)\\
& = \com_0(f_0,g_0) + \com_0((h_0)_L,z_0), \label{eq:com-map-formula-1}
\end{align}
where $z_L=z_R+z_0$ is used in the second equation,
and Lemma \ref{lemma:wn-0-1-commutator} is used in the third equation.
Since we know the explicit formula \eqref{eq:B0andc0} of $\com_0$,
we have obtained one explicit formula of $\com$.
In addition, since the formula \eqref{eq:B0andc0} of $\com_0$ is invariant under any reparameterization of $S^1$,
so is $\com$; see footnote \ref{fn:reprm-inv}.

\begin{figure}[H]
\[\vcenter{\hbox{\includegraphics[width=.8\textwidth]{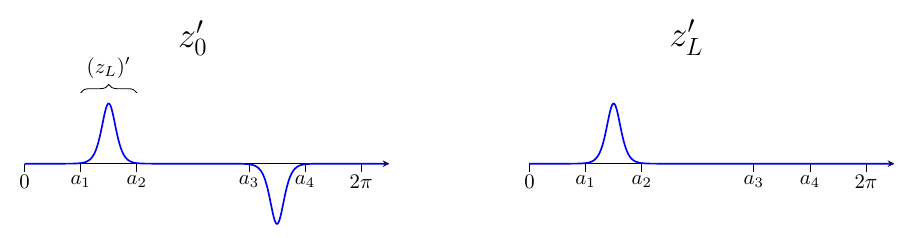}}}\]
\caption{
The graph of functions $z_0'$ and $z_L'$.\label{fig:zz}}
\end{figure}
We can further deform this formula (\ref{eq:com-map-formula-1})
so that the smeared function $(h_0)_L$ and the bump function $z_0$ do not appear as follows,
which will lead to the explicit formula of Theorem \ref{thm:uniqueness-com}.
First, by applying the formula \eqref{eq:B0andc0} of $\com_0$, the first term of (\ref{eq:com-map-formula-1}) can be deformed as
\begin{align}
\com_0(f_0,g_0) & = \frac{k}{2}\int_{S^1}\bigl((f-w_f z_L)(g'-w_g z_L') - (g-w_g z_L)(f'-w_f z_L')\bigr)dx\\
& = \frac{k}{2}\left(\int_0^{2\pi}(fg'-gf')dx - \int_0^{2\pi}h_0z_L'dx + \int_0^{2\pi}h_0'z_Ldx \right),
\end{align}
where we used an easy observation
\begin{align}
h_0 = w_g f - w_f g
\end{align}
in the last equation.\footnote{
Also note that $\int_{S^1}$ is well-defined only when the integrand has winding number zero,
and when we decompose the integration into those with nonzero-winding-number integrands,
we have to specify the domain of integration $[0,2\pi]$ before doing so.
}
On the other hand, by carefully observing the shapes of the functions on each interval $[a_i,a_{i+1}]$ as shown in Fig.~\ref{fig:zz},
we can rewrite the second term of (\ref{eq:com-map-formula-1}) as
\begin{align}
\com_0((h_0)_L,z_0) & = \frac{k}{2}\left( \int_{S^1}h_0z_L'dx - \left( \int_{[0,a_2]}h_0'z_Ldx + \int_{[a_2,a_3]}((h_0)_L)'dx \right)\right)\\
& = \frac{k}{2}\left( \int_{S^1}h_0z_L'dx - \int_{[0,a_2]}h_0'z_Ldx + h_0(a_2) \right).
\end{align}

Finally, by adding these terms, Eq.~(\ref{eq:com-map-formula-1}) becomes
\begin{align}
\com(f,g) & = \frac{k}{2} \left( \int_0^{2\pi}(fg'-gf')dx + \int_{[a_2,2\pi]}h_0'z_Ldx + h_0(a_2) \right)\\
& = \frac{k}{2} \left( \int_0^{2\pi} (fg'-gf')dx + h_0(2\pi) \right)\\
& = \frac{k}{2} \left( \int_0^{2\pi} (fg'-gf')dx +w_gf(0) - w_fg(0) \right). \label{eq:com-map-formula-2}
\end{align}
This is the formula of $\com$ in Theorem \ref{thm:uniqueness-com}.
We can see that $\com(f+1,g)=\com(f,g)+kw_g$,
so $k$ must be an integer from ($\com$-0).

Conversely,
this $\com$ in \eqref{eq:com-map-formula-2} with $k$ integer satisfies all the conditions ($\com$-0)--($\com$-3).
The conditions from ($\com$-0) to ($\com$-2) are obvious.
To see ($\com$-3),
it suffices to show that $\int_0^{2\pi}fg'dx-w_fg(0)$ is an integer if $\supp f \cap \supp g = \varnothing$.
This is easy when $0\not\in\supp g$,
and also can be checked by a straightforward calculation when $0\in\supp g$.
This ends the proof of Theorem \ref{thm:uniqueness-com}.

\section{Computational details of anomaly cocycles}

\subsection{The cocycle condition for bosonic anomaly \texorpdfstring{$\alpha$}{alpha}} 
\label{subsec:cocycleconditionboson}
We apply $\rho_l$ on both sides of \eqref{eq:Bosonkeyequality}, and get 
\begin{equation}
        (\rho_l\rho_g) (u_{h,k}) \rho_l(u_{g,h k}) = e^{2\pi i\alpha(g,h,k)} \rho_l(u_{g,h})  \rho_l(u_{gh,k})\,.
\end{equation}
Using \eqref{eq:localunitaryrepeat} again on the left-hand side (LHS)
\begin{equation}
\begin{split}
     \text{LHS} &= u_{l,g} \rho_{lg}(u_{h, k}) u_{l,g}^{-1} \rho_l(u_{g, h k}) \\
     & = u_{l,g} e^{2\pi i\alpha(l g,h,k)} u_{l g, h} u_{l g h,k} u_{lg,hk}^{-1} u_{l,g}^{-1}  e^{2\pi i\alpha(l, g, hk)} u_{l,g} u_{lg,hk} u_{l,ghk}^{-1}\,,
    \end{split}
\end{equation}
here we need to use \eqref{eq:Bosonkeyequality} twice.
For the right-hand side (RHS), we also use \eqref{eq:Bosonkeyequality} twice and the result is
\begin{equation}
\text{RHS} = e^{2\pi i(\alpha(g, h, k) +\alpha(l, g,h))} u_{l,g} u_{lg,h} u_{l,gh}^{-1} e^{2\pi i\alpha(l, gh, k)} u_{l,g h} u_{l g h,k} u_{l,g h k}^{-1}\,.
\end{equation}
Now, if we look at both sides, all the $u$'s cancel out and we get an equation on $\alpha$
\begin{equation}
    \alpha(lg,h,k)+  \alpha(l,g,hk) = \alpha(g,h,k)+ \alpha(l,g,h)+  \alpha(l,gh,k) \, \pmod \bZ
\end{equation}
which means that $\alpha$ is a $\bR/\bZ$-valued 3-cocycle, i.e.~$\delta \alpha = 0$.

\subsection{The constraint for fermionic anomaly \texorpdfstring{$(\nu,\alpha)$}{(nu, alpha)}}
\label{subsec:constraintfermion}
The computation is similar to that of the bosonic case. We apply $\rho_l$ on both sides of \eqref{eq:keyequality},
making use of \eqref{eq:fermionfusion}.
We find
\begin{equation}
\begin{split}
     \text{LHS} 
     & = (-1)^{\nu(l, g) \nu(h, k)} \lambda_{l,g} e^{2\pi i\alpha(l g,h,k)}  \lambda_{l g,h} \lambda_{l g h,k} (\lambda_{lg,hk})^{-1} (\lambda_{l,g})^{-1} \\
      & \quad \times e^{2\pi i\alpha(l, g, hk)}  \lambda_{l,g} \lambda_{l g, h k} (\lambda_{l, g hk})^{-1}\,
    \end{split}
\end{equation}
and 
\begin{equation}
\begin{split}
     \text{RHS}  & = e^{2\pi i(\alpha(g,h,k)+ \alpha(l, g,h))}\lambda_{l,g} \lambda_{l g,h} (\lambda_{l, g h})^{-1} \\
    & \quad \times e^{2\pi i\alpha(l ,g h, k)} \lambda_{l,g h} \lambda_{l g h,k} (\lambda_{l, g h k})^{-1}\,.
    \end{split}
\end{equation}
Comparing both sides, we find that $\lambda$ all cancel out again, and  we have
\begin{equation}
    \alpha(l, g, h)+ \alpha(l,gh,k)+ \alpha(g, h,k)-\alpha(l g, h, k)-\alpha(l,g, hk) 
    =\frac12 \nu(l, g)  \nu(h,k)  
\end{equation} modulo $\bZ$,
which means 
\begin{equation}
   \delta \alpha = \frac12\nu^2
\end{equation}
as promised.

\subsection{The equivalence relation between fermionic anomalies}
\label{subsec:equivalencefermion}
We drop the superscript $L$, as all operators are from $\cA^L$. The defining equation of $\tilde \alpha$ is 
\begin{equation}  
   \tilde\rho_g (\tilde\lambda_{h,k}) \tilde\lambda_{g,hk} = e^{2\pi i\tilde\alpha(g,h,k)} \tilde\lambda_{g,h} \tilde\lambda_{gh,k}\,,
\end{equation}
and we insert the expression \eqref{eq:fermionicnewconj} and \eqref{eq:utilde} for $\tilde\rho$ and $\tilde\lambda$. The LHS of the above equation is 
\begin{equation}
\begin{split}
     \text{LHS} &= (-1)^{\xi(g)\tilde\nu(h,k)} \Sigma_g \rho_g(\Sigma_h) \rho_g\left(\rho_h(\Sigma_k) \right) \rho_g(\lambda_{h,k})  \rho_g\left( (\Sigma_{hk})^{-1} \right) (\Sigma_{g})^{-1} \\
     &\times \Sigma_{g}  \rho_g(\Sigma_{hk}) \lambda_{g,hk}  (\Sigma_{ghk})^{-1} \,,\\
    \end{split}
\end{equation}
while the RHS reads 
\begin{equation}
     \text{RHS} = e^{2\pi i\tilde\alpha(g,h,k)} \Sigma_{g}  \rho_g(\Sigma_{h}) \lambda_{g,h}  (\Sigma_{gh})^{-1} \Sigma_{gh}  \rho_{gh}(\Sigma_{k}) \lambda_{gh,k}  (\Sigma_{ghk})^{-1} \,.
\end{equation}
Equaling the two expressions yields
\begin{equation}
\begin{split}
   & (-1)^{\xi(g)\tilde\nu(h,k)}  \rho_g\left(\rho_h(\Sigma_k) \right) \rho_g(\lambda_{h,k})   \lambda_{g,hk} = e^{2\pi i\tilde\alpha(g,h,k)}    \lambda_{g,h}  \rho_{gh}(\Sigma_{k}) \lambda_{gh,k} \\
  \Longrightarrow \quad \quad &  (-1)^{\xi(g)\tilde\nu(h,k) + \xi(k) \nu(g,h)}  \lambda_{g,h}\rho_{gh}(\Sigma_k) (\lambda_{g,h})^{-1} e^{2\pi i\alpha(g,h,k)}  \lambda_{g,h} \lambda_{gh,k} \\
 & \qquad\qquad\qquad\qquad\qquad\qquad\qquad\quad \; \; = e^{2\pi i\tilde\alpha(g,h,k)}    \lambda_{g,h}  \rho_{gh}(\Sigma_{k}) \lambda_{gh,k}\\
 \Longrightarrow \quad \quad & \tilde \alpha(g,h,k) = \alpha(g,h,k) + \frac{1}{2}(\xi(g)\tilde\nu(h,k) + \xi(k) \nu(g,h))  \, \pmod \bZ \\
 &\qquad\quad\;\;\; = \alpha(g,h,k) + \frac{1}{2}(\xi(g)(\nu(h,k) + \delta\xi(h,k) )+ \xi(k) \nu(g,h))  \, \pmod \bZ \,,
\end{split}
\end{equation}
where we used \eqref{eq:fermionfusion} and \eqref{eq:keyequality}.

\subsection{Addition formula}
\label{subsec:supertensor}
Consider   two fermionic theories with Hilbert spaces $\cH$ and $\cH'$.
Let $U,V$ act on $\cH$ and $U'$, $V'$ act on $\cH'$.
Denote by $|O|$ be the fermionic parity of operators $O$.
\if0
, and take the tensor product $\cH \otimes \cH'$, then there is a canonical graded tensor product isomorphism 
\begin{equation}
\begin{split}
    \cH \otimes \cH' &\longrightarrow \cH' \otimes \cH \\
    v \otimes v' &\longmapsto (-1)^{|v| |v'|} v' \otimes v \,,
\end{split}
\end{equation}
where $|v|, |v'|$ are fermion parities of the state vectors, and similarly for other cases $\cH \otimes \cH^*$ and $\cH' \otimes \cH'^*$. We can view unitary operators acting on $\cH$ represented as element in $\cH \otimes \cH^*$
\begin{equation}
    U = \sum_{i,j} U_{i j} \,e_i \otimes e^*_j \in \cH \otimes \cH^* \,,
\end{equation}
here we have a formal sum over $i$ and $j$ with respect to the basis $\{e_i\}$ and dual basis $\{e^*_j\}$. Keeping this in mind, the operator algebra on $\cH \otimes \cH'$ satisfies following multiplication rule 
\fi
Let $U\hotimes U'$ denote the combined operation of $U$ and $U'$ on $\cH\otimes \cH'$.
We have an important relation
\begin{equation} \label{eq:fermiontensormulti}
    (U \hotimes U') \cdot (V \hotimes V') = (-1)^{|U'||V|} (U V) \hotimes ( U'  V')\,.
\end{equation}
This additional sign is necessary to make fermionic operators on $\cH$ and $\cH'$ to anticommute,
rather than to commute.
In terms of the ordinary tensor product, this can be achieved by setting \begin{equation}
U\hotimes U' :=  (U (-1)^{ F|U'|}) \otimes U',
\end{equation} where $(-1)^F$  is the fermion parity operator on $\cH$.

We consider $\hat\rho_g(O\hotimes O')  = \rho_g(O) \hotimes \rho'_g(O')$.  
Similarly, we have the fusion operator 
\begin{equation}
    \hat u_{g,h} = u_{g,h} \hotimes {u'}_{g,h} \,,
\end{equation}
from which we immediately see that 
\begin{equation}
    \hat \nu(g,h) := |\hat u_{g,h}| = |u_{g,h}|+ |{u'}_{g,h}| = \nu(g,h) + \nu'(g,h) \, \pmod 2 \,.
\end{equation}
We use Eq.~\eqref{eq:keyequality} to define a new anomaly 3-cochain $\hat \alpha$
\begin{equation}
       \hat\rho_g (\hat\lambda_{h,k}) \cdot \hat\lambda_{g,hk} = e^{2\pi i\hat\alpha(g,h,k)} \hat\lambda_{g,h} \cdot\hat\lambda_{gh,k}\,.
\end{equation}
Plugging in the definitions, we obtain
\begin{equation*}
     \left(  \rho_g (u_{h,k}) \hotimes  \rho'_g({u'}_{h,k}) \right) \cdot \left( u_{g,h k} \hotimes {u'}_{g,h k} \right) = e^{2\pi i\hat\alpha(g,h,k)}  \left(u_{g,h} \hotimes {u'}_{g,h}  \right)\cdot \left( u_{gh,k} \hotimes {u'}_{gh,k}\right) \,.
\end{equation*}
Evaluating both sides using the multiplication rule~\eqref{eq:fermiontensormulti} and the defining properties of $(\nu,\alpha)$ and $(\nu',\alpha')$, we discover the addition formula
\begin{equation}
\begin{split}
        \hat \alpha(g, h, k) &= \alpha(g, h, k) + \alpha'(g, h, k) + \frac{1}{2} (\nu(g,h k) \nu'(h,k) + \nu(gh,k) \nu'(g,h)) \pmod \bZ\, \\
        &= \left(\alpha+\alpha'+\frac12 \nu \cup_1 \nu' \right)(g,h,k) \pmod \bZ \,.
\end{split}
\end{equation}

 \bibliographystyle{ytamsalpha}
 \def\arxivfont{\rm}
 \baselineskip=.95\baselineskip
\bibliography{ref}

\end{document}